\documentclass[11pt]{article}

\usepackage{epsfig,color,amsmath,times}

\usepackage[numbers]{natbib}
\def\nofull{0}

\ifnum\nofull=0
\usepackage{fullpage}
\else
\usepackage{full}
\usepackage{algorithmic}
\fi
\usepackage{enumitem}

  \newtheorem{theorem}{Theorem} \newtheorem{thm}{Theorem}[section]
    \newtheorem{lemma}[thm]{Lemma}
   \newtheorem{corollary}[thm]{Corollary} \newtheorem{observation}[thm]{Observation}
   
   \newtheorem{definition}{Definition} 
   \newtheorem{alg}{Algorithm}

\newcommand\myeq{\mathrel{\overset{\makebox[0pt]{\mbox{\normalfont\tiny\sffamily def}}}{=}}}
\newcommand{\BT}{\begin{thm}} \newcommand{\ET}{\end{thm}}
\newcommand{\BL}{\begin{lemma}} \newcommand{\EL}{\end{lemma}}
\newcommand{\BCM}{\begin{observation}} \newcommand{\ECM}{\end{observation}}
\newcommand{\BD}{\begin{definition}} \newcommand{\ED}{\end{definition}}
\newcommand{\BA}{\begin{alg}} \newcommand{\EA}{\end{alg}}

\newcommand{\BE}{\begin{enumerate}} \newcommand{\EE}{\end{enumerate}}
\newcommand{\BI}{\begin{itemize}} \newcommand{\EI}{\end{itemize}}

\def\FullBox{\hbox{\vrule width 8pt height 8pt depth 0pt}}
\newcommand{\qed}{\;\;\;\FullBox}
\newenvironment{proof}{\noindent{\bf Proof:~~}}{\(\qed\)}

\newcommand{\BEQ}{\begin{equation}} \newcommand{\EEQ}{\end{equation}}
\newcommand{\BEQN}{\begin{eqnarray}}\newcommand{\EEQN}{\end{eqnarray}}
\newcommand{\BPF}{\begin{proof}} \newcommand {\EPF}{\end{proof}}
\newenvironment{proofof}[1]{\noindent{\bf Proof of {#1}:~~}}{\(\qed\)}
\newcommand{\BPFOF}{\begin{proofof}} \newcommand {\EPFOF}{\end{proofof}}

\newcommand{\blue}[1]{\textcolor{black}{#1}}

\newcommand{\poly}{{\rm poly}}

\begin{document}

\begin{titlepage}
\title{Approximating the noise sensitivity of a monotone Boolean function}

\author{
	Ronitt Rubinfeld \thanks{
		CSAIL at MIT, and the
		Blavatnik School of Computer Science at Tel Aviv University,
		{\tt  ronitt@csail.mit.edu}.
		Ronitt Rubinfeld's research was supported by  NSF grants CCF-1650733,
		CCF-1733808, IIS-1741137 and CCF-1740751.}
	\and
	Arsen Vasilyan\thanks{	
		CSAIL at MIT,
		{\tt  vasilyan@mit.edu}.
		Arsen Vasilyan's research was supported by the EECS SuperUROP program,
		the MIT Summer UROP program and 
		the DeFlorez Endowment Fund.}
}

\maketitle

\begin{abstract}

\sloppy
The {\em noise sensitivity} of a Boolean function $f: \{0,1\}^n \rightarrow \{0,1\}$ is one of its fundamental properties. A function of a positive noise parameter $\delta$, it is denoted as $NS_{\delta}[f]$. 
Here we study the algorithmic problem of approximating it for monotone $f$, 
such that $NS_{\delta}[f] \geq 1/n^{C}$ for constant $C$,  and where $\delta$ satisfies
$1/n \leq \delta \leq 1/2$.
For such $f$ and $\delta$, we give a randomized algorithm performing $O\left(\frac{\min(1,\sqrt{n} \delta \log^{1.5} n) }{NS_{\delta}[f]} \poly\left(\frac{1}{\epsilon}\right)\right)$ queries and approximating $NS_{\delta}[f]$ to within a multiplicative factor of $(1\pm \epsilon)$. 
Given the same constraints on $f$ and $\delta$, we also prove a lower bound of $\Omega\left(\frac{\min(1,\sqrt{n} \delta)}{NS_{\delta}[f] \cdot n^{\xi}}\right)$ on the query complexity of any algorithm that approximates $NS_{\delta}[f]$ to within any constant factor, where $\xi$ can be any positive constant.  
Thus, our algorithm's query complexity is close to optimal in terms of its dependence on $n$.

We introduce a novel {\em descending-ascending view} of noise sensitivity, and use it as a central tool for the analysis of our algorithm. 
To prove lower bounds on query complexity, we develop a technique that 
reduces computational questions about query complexity to combinatorial 
questions about the existence of ``thin" functions with certain properties. 
The existence of such ``thin" functions is proved using the probabilistic method. 
These techniques also yield previously unknown lower bounds on the query complexity of approximating other fundamental properties of Boolean functions: 
the {\em total influence} and the {\em bias}.
\end{abstract}
\noindent

\end{titlepage}

\section{Introduction}
Noise sensitivity is a property of any Boolean function $f: \{0,1\}^n \rightarrow \{0,1\}$ 
defined as follows: First, pick $x = \{x_1,\ldots,x_n\}$ uniformly at random from $\{0,1\}^n$, then pick $z$ by flipping each $x_i$ with probability $\delta$. 
Here $\delta$, the noise parameter, is a given positive constant
no greater than $1/2$ (and at least $1/n$ in the interesting cases). 
The noise sensitivity of $f$, denoted as $NS_{\delta}[f]$, 
equals the probability that $f(x) \neq f(z)$. This definition was first explicitly given by Benjamini, Kalai and Schramm in \cite{num7}.
Noise sensitivity has been the focus of multiple papers:~\cite{num7,num20,num4,num8,num3,num16,num9,num6}. It has been applied to learning theory ~\cite{num17,num20,num4,
num2,num3,num5,num11}, property testing ~\cite{num19,num18}, hardness of approximation ~\cite{num10,num13}, hardness amplification ~\cite{num12}, theoretical economics and political science\cite{num8}, combinatorics ~\cite{num7,num16}, distributed computing ~\cite{num14} and differential privacy ~\cite{num20}.
Multiple properties and applications of noise sensitivity are summarized in \cite{num1} and \cite{num15}. 

In this work, we study the algorithmic question of \emph{approximating the noise sensitivity} of a function to which we are given oracle access. It can be shown that standard sampling 
techniques require $O \left(\frac{1}{NS_{\delta}[f] \epsilon^2}\right)$ 
queries to get a $(1+\epsilon)$-multiplicative approximation for $NS_{\delta}[f]$. In Appendix B, we show that this is optimal for a wide range of parameters of the problem. Specifically, it cannot be improved by more than a constant when $\epsilon$ is a sufficiently small constant, $\delta$ satisfies $1/n \leq \delta \leq 1/2$ and $NS_{\delta}[f]$ satisfies $\Omega \left(\frac{1}{2^n}\right) \leq NS_{\delta}[f] \leq O(1)$. Therefore, we focus on estimating the noise sensitivity of \emph{monotone} functions, which form an important subclass of Boolean functions.  

Noise sensitivity is closely connected to the total influence (henceforth just influence) of a Boolean function \cite{num1,num15}. 
Influence, denoted by $I[f]$, is defined as $n$ times the probability that $f(x) \neq f(x^{\oplus i})$, 
where $x$ is chosen uniformly at random, as for noise sensitivity, 
and $x^{\oplus i}$ is formed from $x$ by 
flipping a random index
(this latter probability 
is sometimes referred to as the {\em average sensitivity}). 
Despite their close connections, 
the noise sensitivity of a function can be quite different from its influence. 
For instance, for the parity function of all $n$ bits, 
the influence is $n$, but the noise sensitivity is 
$\frac{1}{2}(1-(1-2\delta)^n)$ (such disparities
also hold for monotone functions, 
see for example the discussion of influence and noise sensitivity of the majority function in \cite{num15}).

The approximation of parameters of monotone Boolean functions was
previously studied by \cite{num21, num26}, where they considered the 
question of how fast one can 
approximate the influence of a monotone function $f$ given oracle access to $f$.
It was shown that one can approximate the influence of a monotone function $f$ with only $\tilde{O} \left(\frac{\sqrt{n}}{I[f] \text{poly}(\epsilon)} \right)$ queries, which for constant $\epsilon$ beats the standard sampling algorithm by a factor of $\sqrt{n}$, ignoring logarithmic factors. 

Noise sensitivity is also closely related to the notion of noise stability (see for instance \cite{num15}). The noise stability of a Boolean function $f$ depends on a parameter $\rho$ and is denoted by $\text{Stab}_{\rho}[f]$. A well known formula connects it with noise sensitivity $NS_{\delta}[f]=\frac{1}{2}(1-\text{Stab}_{1-2 \delta}[f])$. 
This implies that by obtaining an approximation for $NS_{\delta}[f]$, 
one also achieves an approximation for the gap between $\text{Stab}_{1-2 \delta}[f]$ and one.

\subsection{Results}
Our main algorithmic result is the following:
\begin{theorem}
\label{main theorem}
Let $\delta$ be a parameter satisfying:
$$
\frac{1}{n}
\leq
\delta
\leq
\frac{1}{\sqrt{n} \log n}
$$
Suppose, $f:\{0,1\}^n \rightarrow \{0,1\}$ is a monotone function and $NS_\delta[f] \geq \frac{1}{n^{C}}$ for some constant $C$. 

Then, there is an algorithm that outputs an approximation to $NS_{\delta}[f]$ to within a multiplicative factor of $(1 \pm \epsilon)$,
with success probability at least $2/3$.   In expectation, the algorithm
makes $O\left(  \frac{\sqrt{n} \delta \log^{1.5} n}{NS_\delta[f] \epsilon^3} \right)$
  queries to the function. Additionally, it runs in time polynomial in $n$.
\end{theorem}

Note that computing noise-sensitivity using standard sampling\footnote{Standard sampling refers to the algorithm that picks 
$O \left(\frac{1}{NS_{\delta}[f] \epsilon^2}\right)$ pairs 
$x$ and $z$ as in the definition of noise sensitivity and computes the fraction of pairs
for which $f(x) \neq f(z)$.} requires $O \left(\frac{1}{NS_\delta[f] \epsilon^2} \right)$ samples.  Therefore, for a constant $\epsilon$, we have the most dramatic improvement if $\delta=\frac{1}{n}$, 
in which case, ignoring constant and logarithmic factors, 
our algorithm outperforms standard sampling by a factor of $\sqrt{n}$.

As in \cite{num21}
\footnote{In the following, we discuss how we build on \cite{num21}. We discuss the relation with \cite{num26} in Subsection \ref{possibilities}.}
, our algorithm requires that the noise sensitivity 
of the input function $f$ is larger than a specific threshold $1/n^C$. 
Our algorithm is not sensitive to the value of $C$ as long as it is a constant, 
and we think of $1/n^C$ as a rough initial lower bound known in advance.

We next give lower bounds for approximating three different quantities
of monotone Boolean functions:  the
bias, the influence and the noise sensitivity.  
A priori, it is not clear what kind of lower bounds one could hope for.
Indeed, determining whether a given function is the all-zeros function requires $\Omega(2^n)$ queries in the general function setting, but only $1$ query (of the all-ones input), if the function is promised to be monotone.
Nevertheless, we show that such a dramatic improvement for approximating
these quantities is not possible.


For monotone functions, we are not aware of previous
lower bounds on approximating the bias or noise sensitivity.
Our lower bound on approximating influence
is not comparable to the lower bounds in \cite{num21}, as we will elaborate
shortly.  

We now state our lower bound for approximating the noise
sensitivity: 
\begin{theorem} 
\label{2-1.3}
Suppose $n$ is a sufficiently large integer, $\delta$ 
satisfies $1/n \leq \delta \leq 1/2$, $C_1$ and $C_2$ are constants satisfying $C_1 - 1 > C_2 \geq 0$. Then,
given a monotone 
function $f:\{0,1\}^n \rightarrow \{0,1\}$,  one needs at least $ 
		\Omega\left( \frac{n^{C_2}}{e^{\sqrt{C_1 \log n/2}}} \right)$ 
queries to reliably distinguish\footnote{Here and everywhere else, to reliably distinguish means to distinguish with probability at least $2/3$.}  between the following two cases:
(i) $f$ has noise sensitivity between $\Omega(1/n^{C_1+1})$ and $O(1/n^{C_1})$ and 
(ii) $f$ has  noise sensitivity larger than $\Omega(\min(1,\delta \sqrt{n})/n^{C_2})$.
\end{theorem}

\noindent
\textbf{Remark:} For any positive constant $\xi$, we have that
$e^{\sqrt{C_1 \log n/2}} \leq n^{\xi}$.

\noindent
\textbf{Remark:} 
The range of the  parameter $\delta$ can be
divided into two regions of interest. 
In the region $1/n \leq \delta \leq 1/(\sqrt{n} \log n)$, the algorithm from Theorem \ref{main theorem} can distinguish the two cases above with only $\tilde{O}(n^{C_2})$ queries. Therefore its query complexity is optimal up to a factor of $\tilde{O}(e^{\sqrt{C_1 \log n/2}})$.
Similarly, in the region $1/(\sqrt{n} \log n) \leq \delta \leq 1/2$, the standard sampling algorithm can distinguish the two distributions above with only $\tilde{O}(n^{C_2})$ queries. Therefore in this region of interest, standard sampling is optimal up to a factor of $\tilde{O}(e^{\sqrt{C_1 \log n/2}})$.



We define the {\em bias} of a Boolean function as $B[f] \myeq Pr[f(x)=1]$, where $x$ is chosen uniformly at random from $\{0,1\}^n$. It is arguably the most basic property of a Boolean function, so we consider the question of how quickly it can be approximated for monotone functions. 
To approximate the bias using standard sampling, 
one needs $O(1/(B[f] \epsilon^2))$ queries. 
We obtain a lower bound for approximating it similar to the previous theorem:

\begin{theorem} 
\label{2-1.1}
Suppose $n$ is a sufficiently large integer, $C_1$ and $C_2$ are constants satisfying $C_1 - 1 > C_2 \geq 0$. Then
given a monotone
function $f$\blue{$:\{0,1\}^n \rightarrow \{0,1\}$},  one needs at least $ 
		\Omega\left( \frac{n^{C_2}}{e^{\sqrt{C_1 \log n/2}}} \right)$ 
queries to \blue{reliably} distinguish  between the following two cases:
(i) $f$ has bias of $\Theta(1/n^{C_1})$ 
(ii) $f$ has  bias larger than $\Omega(1/n^{C_2})$.
\end{theorem}

Finally we prove a lower bound for approximating influence:
\begin{theorem} 
\label{2-1.2}
Suppose $n$ is a sufficiently large integer, $C_1$ and $C_2$ are constants satisfying $C_1 - 1 > C_2 \geq 0$. Then
given a monotone
function $f$\blue{$:\{0,1\}^n \rightarrow \{0,1\}$}, one needs at least $ 
		\Omega\left( \frac{n^{C_2}}{e^{\sqrt{C_1 \log n/2}}} \right)$ 
queries to \blue{reliably} distinguish  between the following two cases:
(i) $f$ has influence between $\Omega(1/n^{C_1})$ and $O(n/n^{C_1})$
(ii) $f$ has  influence larger than $\Omega(\sqrt{n}/n^{C_2})$.
\end{theorem}
This gives us a new sense in which the algorithm family in \cite{num21} is close to optimal, because for a function $f$ with influence $\Omega(\sqrt{n}/n^{C_2})$ this algorithm makes $\tilde{O}(n^{C_2})$ queries to estimate the influence up to any constant factor.

Our lower bound is incomparable to the lower bound in \cite{num21}, 
which makes the stronger requirement that $I[f] \geq \Omega(1)$, 
but gives a bound that is only a polylogarithmic 
factor smaller than the runtime of the algorithm in \cite{num21}. 
There are many possibilities 
for algorithmic bounds
that were compatible with the lower bound in \cite{num21}, 
but are eliminated with our lower bound. 
For instance, prior to this work, it was conceivable that an algorithm 
making as little as $O(\sqrt{n})$ queries could give a constant 
factor approximation to the influence of \textbf{any} monotone input function 
whatsoever. Our lower bound shows that not only is this impossible, 
no algorithm that makes $O(n^{C_2})$ 
queries for any constant $C_2$ can accomplish this either.

\subsection{Algorithm overview}
Here, we give the algorithm in Theorem \ref{main theorem} together with the subroutines it uses. Additionally, we give
an informal overview of the proof of correctness and the analysis of run-time and query complexity, which are presented in Section \ref{NS algo}. 

First of all, with $x$ and $z$ described as above, using a standard pairing argument we argue that $NS_{\delta}[f]=2 \cdot Pr[f(x)=1 \land f(z)=0]$. In other words, we can focus only on the case when the value of the function flips from one to zero. 

We introduce the {\em descending-ascending view of noise sensitivity}
(described in Subsection \ref{subsec:view}), which, roughly speaking, 
views the noise process as decomposed into a first \blue{phase} that operates only on the locations in $x$ that are $1$, and a second
\blue{phase} that operates only on the locations in $x$ that are set
to $0$. Formally, we define:  

\hfill \newline
\noindent\fbox{
    \parbox{\textwidth}{
\textbf{Process $D$}
\setlist{nolistsep}
\begin{itemize}[noitemsep]
\item Pick $x$ uniformly at random from $\{0,1\}^n$. Let $S_0$ be the set of indexes $i$ for which $x_i=0$, and conversely let $S_1$ be the rest of indexes.
\item \textbf{Phase 1:} go through all the indexes in $S_1$ in a random order, and flip each with probability $\delta$. Form the descending path $P_1$ from all the intermediate results. Call the endpoint $y$.
\item \textbf{Phase 2:} start at $y$, and flip each index in $S_0$ with probability $\delta$. As before, all the intermediate results form an ascending path $P_2$, which ends in $z$.
\end{itemize}
}}
\hfill \newline
This process gives us a path from
$x$ to $z$ that can be decomposed into two segments, such that the first part,
$P_1$, descends in the hypercube, and the second part $P_2$ ascends 
in the hypercube.  




Since $f$ is monotone, for $f(x)=1 $ and $f(z)=0$ to be the case, it is necessary, though not sufficient, that $f(x)=1$ and $f(y)=0$, which happens whenever $P_1$ hits an influential edge. 
Therefore we break the task of estimating the probability of $f(x) \neq f(z)$ into computing the product of:
\begin{itemize}
\item 
The probability that $P_1$ hits an influential edge, specifically,
the probability that $f(x)=1$ and $f(y)=0$, which we refer to as $p_A$. 
\item 
The probability that $P_2$ does not hit any influential edge,
 given that $P_1$ hits an influential edge:  specifically,
the probability that given $f(x)=1$ and $f(y)=0$, it is the case that $f(z)=0$.
We refer to this probability as $p_B$.

\end{itemize} 
The above informal definitions of $p_A$ and $p_B$ ignore some technical complications.
Specifically, the impact of certain ``bad events" is considered in our analysis. We redefine $p_A$ and $p_B$ precisely in Subsection~\ref{pa,pb}. 

To define those bad events, we use the following two values, which we reference in our algorithms: $t_1$ and $t_2$. Informally, $t_1$ and $t_2$ have the following intuitive meaning. A typical vertex $x$ of the hypercube has Hamming weight $L(x)$ between $n/2-t_1$ and $n/2+t_1$. A typical Phase 1 path from process $D$ will have length at most $t_2$. To achieve this, we assign $t_1\myeq \eta_1\sqrt{n \log n} $ and $t_2 \myeq n \delta (1+3\eta_2 \log n)$, where $\eta_1$ and $\eta_2$ are certain constants. We also define $M$ to be the set of edges $e=(v_1, v_2)$, for which both $L(v_1)$ and $L(v_2)$ are between and $n/2-t_1$ and $n/2+t_1$. Again, most of the edges in the hypercube are in $M$. 

The utility of this ascending-descending view is that given\footnote{Our analysis requires that only $\delta \leq 1/(\sqrt{n} \log n)$ as in the statement of Theorem \ref{main theorem}, however the analysis of the case $1/(\sqrt{n} \log^{2} n) \leq \delta \leq 1/(\sqrt{n} \log n)$ involves technical subtleties that
we ignore here.} $\delta \leq 1/(\sqrt{n} \log^{2} n)$, 
it is the case that $t_2$ will be shorter than $O(\sqrt{n}/\log n)$. Therefore, typically, the path $P_1$ is also shorter than $O(\sqrt{n} /\log n)$.
Similar short descending paths on the hypercube have been studied before: 
In \cite{num21}, paths of such lengths were used to 
estimate the number of influential edges by analyzing the probability
that a path would hit such an edge.  One useful insight given
by \cite{num21} is that 
the probability of hitting almost every single influential 
edge is roughly the same. 

However, the results in \cite{num21} cannot be immediately applied to analyze $P_1$, 
because (i) $P_1$ does not have a fixed length, 
but rather its lengths form a probability distribution, (ii) this probability distribution also depends on the starting point $x$ of $P_1$. 
We build upon the techniques in \cite{num21} to overcome these difficulties, 
and prove that again, roughly speaking, 
for almost every single influential edge, 
the probability that $P_1$ hits it depends very little on the location
of the edge, and our proof also computes this probability. 
This allows us to prove that $p_A \approx \delta I[f]/2$. Then, using the algorithm in \cite{num21} to estimate $I[f]$, we estimate $p_A$. 

Regarding $p_B$, we estimate it by approximately sampling paths $P_1$ and $P_2$ that would arise from process $D$,  
conditioned on that $P_1$ hits an influential edge. To that end, we first sample an influential edge $e$ that $P_1$ hits. Since $P_1$ hits almost every single influential edge with roughly the same probability, 
we do it by sampling $e$ approximately uniformly from among influential edges. 
For the latter task, we build upon the result in 
\cite{num21} as follows: 
As we have already mentioned, the algorithm in \cite{num21} samples descending paths of a fixed length to estimate the influence.  For those paths
that start at an $x$ for which $f(x)=1$ and end at a $z$ for which
$f(z)=0$, we add a binary search step in order to locate the influential 
edge $e$ that was hit by the path. \blue{Thus, we have the following algorithm:}

\hfill \newline
\noindent\fbox{
    \parbox{\textwidth}{
\textbf{Algorithm $\mathcal{A}$} (given oracle access to a monotone function $f:\{0,1\}^n \rightarrow \{0,1\}$ and a parameter $\epsilon$)
\setlist{nolistsep}
\begin{enumerate}[noitemsep]
\item Assign $w=\frac{\epsilon}{3100 \eta_1} \sqrt{\frac{n}{\log n}}$ 
\item Pick $x$ uniformly at random from $\{0,1\}^n$.
\item Perform a descending walk $P_1$ downwards in the hypercube starting at $x$. Stop at a vertex $y$ either after $w$ steps, or if you hit the all-zeros vertex. 
\item If $f(x)=f(y)$ output FAIL.
\item If $f(x) \neq f(y)$ perform a binary search on the path $P_1$ and find an influential edge $e_{inf}$.
\item If $e_{inf} \in M$ return $e_{inf}$. Otherwise output FAIL.
\end{enumerate}
}}
\hfill \newline
Finally, once we have obtained a roughly uniformly random influential edge $e$, we sample a path $P_1$ from among those that hit it. Interestingly, we show that this can be accomplished by a simple exponential time algorithm that makes no queries to $f$. However, the constraint on the run-time of our algorithm forces us to follow a different approach:

An obvious way to try to quickly sample such a path is to perform two random walks of lengths $w_1$ and $w_2$ in opposite directions from the endpoints of the edge, and then concatenate them into one path. However, to do this, one needs to somehow sample the lengths $w_1$ and $w_2$. This problem is not trivial, since longer descending paths are more likely to hit an influential edge, which biases the distribution of the path lengths towards longer ones. 

To generate $w_1$ and $w_2$ according to the proper distribution, we first sample a path $P_1$ hitting any edge at the same {\em layer}\footnote{We say that edges $e_1$ and $e_2$ are on the same layer if and only if their endpoints have the same Hamming weights. We denote the layer an edge $e$ belongs to as $\Lambda_e$.} $\Lambda_e$ as $e$.
We accomplish this by designing an algorithm that uses rejection sampling. The algorithm samples short descending paths from some conveniently chosen distribution, until it gets a path hitting the desired layer. Specifically, the algorithm is the following (recall that we use $L(x)$ to denote the Hamming weight of $x$ which equals to the number of indices $i$ on which $x_i=1$, and we use the symbol $\Lambda_e$ to denote the whole layer of edges that have the same endpoint levels as $e$):

\hfill \newline
\noindent\fbox{
    \parbox{\textwidth}{
\textbf{Algorithm $\mathcal{W}$} (given an edge $e\myeq(v_1, v_2)$ so $v_1 \preceq v_2$)
\setlist{nolistsep}
\begin{enumerate}[noitemsep]
\item Pick an integer $l$ uniformly at random among the integers in $[L(v_1) ,L(v_1)+t_2-1]$. Pick a vertex $x$ randomly at level $l$.
\item As in phase 1 of the noise sensitivity process, traverse in random order through the indices of $x$ and with probability $\delta$, flip each index that equals to one to zero. The intermediate results form a path $P_1$, and we call its endpoint $y$.
\item If $P_1$ does not intersect $\Lambda_e$ go to step $1$.
\item Otherwise, output $w_1=L(x)-L(v_1)$ and $w_2=L(v_2)-L(y)$.
\end{enumerate}
}}
\hfill \newline
Recall that $t_2$ has a technical role and is defined to be equal 
$n \delta (1+3\eta_2 \log n)$, where $\eta_2$ is a certain constant. $t_2$ is chosen to be
long enough that it is longer than most paths $P_1$, but short enough to make the sampling in $\mathcal{W}$ efficient. Since the algorithm involves short descending paths, we analyze this algorithm building upon the techniques we used to approximate $p_A$. 

\blue{
After obtaining a random path going through the same layer as $e$, we show how to transform it using the symmetries of the hypercube, into a a random path $P_1$ going through $e$ itself. Additionally, given the endpoint of $P_1$, we sample the path $P_2$ just as in the process $D$.
Formally, we use the following algorithm:}

\hfill \newline
\noindent\fbox{
    \parbox{\textwidth}{
\setlist{nolistsep}
\textbf{Algorithm $\mathcal{B}$} (given an influential edge $e$)
\begin{enumerate}[noitemsep] 
\item Use $\mathcal{W}(e)$ to sample $w_1$ and $w_2$.
\item Perform an ascending random walk of length $w_1$ starting at $v_1$ and call its endpoint $x$. Similarly, perform a descending random walk starting at $v_2$ of length $w_2$, call its endpoint $y$.
\item Define $P_1$ as the descending path that results between $x$ and $y$ by concatenating the two paths from above, oriented appropriately, and edge $e$. 
\item Define $P_2$ just as in phase 2 of our process starting at $y$. Consider in random order all the zero indices $y$ has in common with $x$ and flip each with probability $\delta$.
\item Return $P_1$ ,$P_2$, $x$, $y$ and $z$. 
\end{enumerate}
}}
\hfill \newline
\blue{
We then use sampling to estimate which fraction of the paths $P_2$ continuing these $P_1$ paths 
does not hit an influential edge. This allows us to estimate $p_B$, which, combined with our estimate for $p_A$, gives us an approximation for $NS_{\delta}[f]$.
}
Formally, we put all the previously defined subroutines together as follows:

\hfill \newline
\noindent\fbox{
    \parbox{\textwidth}{
\textbf{Algorithm for estimating noise sensitivity.} (given oracle access to a monotone function $f$\blue{$:\{0,1\}^n \rightarrow \{0,1\}$})
\setlist{nolistsep}
\begin{enumerate}[noitemsep] 
\item \blue{Using the algorithm from \cite{num21} as described in Theorem \ref{estimate influence}, compute an approximation to the influence of $f$ to within a multiplicative factor of ($1 \pm \epsilon/33$)}. This gives us $\tilde{I}$. 
\item Compute $\tilde{p}_A:=\delta \tilde{I}/2$.
\item Initialize $\alpha:=0$ and $\beta:=0$. Repeat the following until $\alpha=\frac{768 \ln 200}{\epsilon^2}$.
\begin{itemize}
\item Use algorithm $\mathcal{A}$ from Lemma \ref{lemma edge sampler} repeatedly to successfully sample an edge $e$.
\item From Lemma \ref{lemma 3} use the algorithm $\mathcal{B}$, giving it $e$ as input, and sample $P_1$, $P_2$, $x$, $y$ and $z$.
\item If it is the case that $f(x)=1$ and $f(z)=0$, then $\alpha:=\alpha+1$.
\item $\beta:=\beta+1$. 
\end{itemize}
\item Set $\tilde{p}_B=\frac{\alpha}{\beta}$.
\item Return $2 \tilde{p}_A \tilde{p}_B$.
\end{enumerate}
}}
\subsection{Lower bound techniques}
We use the same technique to lower bound algorithms
which approximate any of the following three quantities: 
the noise sensitivity, influence and bias. 

For concreteness, let us first focus on approximating the bias.
Recall that one can distinguish
the case where the bias is $0$ from the bias being $1/2^n$ using a single
query.   Nevertheless, we show that for the most part, no
algorithm for estimating
the bias can do much better than the random sampling approach.

We construct two probability distributions $D_1^B$ and $D_2^B$ 
that are relatively hard to distinguish but have drastically different biases. 
To create them, we use a special monotone function $F^B$, which we will explain later how to obtain. 
We pick a function from $D_2^B$ by taking $F^B$, randomly permuting the indices of its input, and finally ``truncating" it by setting it to one on all values $x$ for on levels higher than $l_0$, where $l_0$ is some fixed threshold.     

We form $D_1^B$ even more simply. We take the all-zeros function and truncate it at the same threshold $l_0$. The threshold $l_0$ is chosen in a way that this function in $D_1^B$ has a sufficiently small bias. 

The purpose of truncation is to prevent a distinguisher from gaining information by accessing the values of the function on the high-level vertices of the hypercube. Indeed, if there was no truncation, one could tell whether they have access to the all-zeros function by simply querying it on the all-ones input. Since $F^B$ is monotone, if it equals to one on at least one input, then it has to equal one on the all-ones input.

The proof has two main lemmas: The first one is computational and says that if $F^B$ is ``thin" 
then $D_1^B$ and $D_2^B$ are hard to reliably distinguish. 
We say that a function is ``thin" if it equals to one on only a small fraction of the points on the level $l_0$.
To prove the first lemma, we show that one could transform any adaptive algorithm for distinguishing $D_1^B$ from $D_2^B$ into an algorithm that is just as effective, is non-adaptive and queries points only on the layer $l_0$. 

To show this, we observe that, because of truncation, distinguishing a function 
in $D_2^{B}$ from a function in $D_1^{B}$ is in a certain sense equivalent to finding 
a point with level at most $l_0$ on which the given function evaluates to one. 
We argue that for this setting, adaptivity does not help. 
 Additionally, if $x \preceq y$ and both of them have levels at most $l_0$ then, since $f$ is monotone, $f(x)=1$ implies that $f(y)=1$ (but not necessarily the other way around). Therefore, for finding a point on which the function evaluates to one, it is never more useful to query $x$ instead of $y$.

Once we prove that no algorithm can do better than a non-adaptive algorithm that only queries points on the level $l_0$, we use a simple union bound to show that any such algorithm cannot be very effective for distinguishing our distributions. 

Finally, we need to show that there exist functions that are ``thin" and simultaneously have a high bias. This is a purely combinatorial question and is proven in our second main lemma. We build upon Talagrand random functions that were first introduced in \cite{num25}. In \cite{num9} it was shown that they are very sensitive to noise, which was applied for property testing lower bounds \cite{num18}. A Talagrand random DNF consists of $2^{\sqrt{n}}$ clauses of $\sqrt{n}$ indices chosen randomly with replacement. We modify this construction by picking the indices without replacement and generalize it by picking $2^{\sqrt{n}}/n^{C_2}$ clauses, where $C_2$ is a non-negative constant. We show that these functions are ``thin", so they are appropriate for our lower bound technique. 

``Thinness" allows us to conclude that $D_1^B$ and $D_2^B$ are 
hard to distinguish from each other. 
We then prove that they have drastically different biases. We do 
the latter by employing the probabilistic method and showing 
that in expectation our random function has a large enough bias. 
We handle influence and noise sensitivity analogously, 
specifically by showing that that as we pick fewer clauses, 
the expected influence and noise sensitivity decrease proportionally. 
We prove this by dividing the points where one of these random functions equals to one into two regions: (i) the region where only one clause is true and (ii) a region where more one clause is true. Roughly speaking, we show that the contribution from the points in (i) is sufficient to obtain a good lower bound on the influence and noise sensitivity.

\subsection{Possibilities of improvement?}
\label{possibilities}
In \cite{num26} (which is the journal version of \cite{num21}), it was shown that using the chain decomposition of the hypercube, one can improve the run-time of the algorithm to $O\left( \frac{\sqrt{n}}{\epsilon^2 I[f]}\right)$ and also improve the required lower bound on $I[f]$ to be $I[f]\geq \exp(-c_1 \epsilon^2 n+c_2 \log(n/\epsilon))$ for some constant $c_1$ and $c_2$ (it was $I[f] \geq 1/n^{C}$ for any constant $C$ in \cite{num21}). Additionally, the algorithm itself was considerably simplified.

A hope is that techniques based on the chain decomposition could help improve the algorithm in Theorem \ref{main theorem}. However, it is not clear how to
generalize our approach to use these techniques, since the ascending-descending view is a natural way to express noise sensitivity in terms of random walks, and it is not obvious whether one can replace these walks with chains of the hypercube.

\section{Preliminaries}
\subsection{Definitions}
\subsubsection{Fundamental definitions and lemmas pertaining to the hypercube.}
\begin{definition}
We refer to the poset over $\{0,1\}^n$ as the \textbf{$n$-dimensional hypercube}, 
viewing the domain as vertices of a graph, in which 
two vertices are connected by an edge if and only if the corresponding elements of $\{0,1\}^n$ differ in precisely one index.
For $x=(x_1,\ldots,x_n)$ 
and $y=(y_1,\ldots,y_n)$ in $\{0,1\}^n$, 
we say that $x \preceq y$ if and only if for all $i$ in $[n]$ it is the
case that $x_i \leq y_i$.
\end{definition}
\begin{definition}
The \textbf{level of a vertex} $x$ on the hypercube is the hamming weight of $x$, or in other words number of $1$-s in $x$. We denote it by $L(x)$.
\end{definition}
We define the set of edges that are in the same ``layer" of the hypercube as a given edge:
\begin{definition}
For an arbitrary edge $e$ suppose $e=(v_1,v_2)$ and $v_2 \preceq v_1$. We denote $\Lambda_{e}$ to be the set of all edges $e'=(v_1', v_2')$, so that $L(v_1) = L(v_1')$ and $L(v_2) = L(v_2')$. 
\end{definition}
The size of $\Lambda_{e}$ is $L(v_1){{n}\choose{L(v_1)}}$. The concept of $\Lambda_e$ will be useful because we will deal with paths that are symmetric with respect to change of coordinates, and these have an equal probability of hitting any edge in $\Lambda_e$.
\begin{definition}
\blue{
As we view the hypercube as a graph, we will often refer to paths on it. By referring to a path $P$ we will, depending on the context, refer to its set of vertices or edges. We call the path \textbf{descending} if for every pair of consecutive vertices $v_i$ and $v_{i+1}$, it is the case that $v_{i+1} \prec v_{i}$. Conversely, if the opposite holds and $v_{i} \prec v_{i+1}$, we call the path \textbf{ascending}. We consider an empty path to be vacuously both ascending and descending. We define the length of a path to be the number of edges in it, and denote it by $|P|$. We say we \textbf{take a descending random walk of length $w$ starting at $x$}, if we pick a uniformly random descending path of length $w$ starting at $x$.
}
\end{definition}

Descending random walks over the hyper-cube were used in an essential
way in \cite{num21} and were central for the recent advances in monotonicity testing algorithms \cite{num22, num23, num24}.

\begin{lemma}[Hypercube Continuity Lemma]
\label{continuity}

Suppose \blue{$n$ is a sufficiently large positive integer,} $C_1$ is a constant and we are given $l_1$ and $l_2$ satisfying:

$$
\frac{n}{2}
-\sqrt{C_1 n \log(n)}
\leq l_1
\leq l_2
\leq
\frac{n}{2}
+\sqrt{C_1 n \log(n)}
$$

If we denote $C_2 \myeq \frac{1 }{10 \sqrt{C_1}}$, then for any $\xi$ satisfying $0 \leq \xi \leq 1$,  if it is the case that
$
l_2 - l_1 \leq C_2 \xi \sqrt{\frac{n}{\log(n)}}
$, then, for large enough $n$, it is the case that $
1-\xi
\leq
\frac{{{n}\choose{l_1}}}{{{n}\choose{l_2}}}
\leq 1+\xi
$
\end{lemma}
\begin{proof}
See Appendix C, Subsection \ref{proof 1}.
\end{proof}

\subsubsection{Fundamental definitions pertaining to Boolean functions}
We define monotone functions over the $n$-dimensional
hypercube:
\begin{definition}
Let $f$ be a function $\{0,1\}^n \rightarrow \{0,1\}^n$. We say that $f$ is \textbf{monotone} if for any $x$ and $y$ in $\{0,1\}^n$, $x \preceq y$ implies that $f(x) \leq f(y)$
\end{definition}
Influential edges, and influence of a function are defined as follows:
\begin{definition}
An edge $(x,y)$ in the hypercube is called \textbf{influential} if $f(x) \neq f(y)$. Additionally, we denote the set of all influential edges in the hypercube as $E_I$. 
\end{definition}

\begin{definition}
The \textbf{influence} of function $f$\blue{$:\{0,1\}^n \rightarrow \{0,1\}$}, denoted by $I[f]$ is:

$$
I[f] \myeq n \cdot Pr_{x \in_R \{0,1\}^n, i \in_R [n]}
[f(x) \neq f(x^{\oplus i})]
$$
Where $x^{\oplus i}$ is $x$ with its $i$-th bit flipped.\footnote{
We use the symbol $\in_R$ to denote, depending on the type of object the symbol is followed by: (i) Picking a random element from a probability distribution. (ii) Picking a uniformly random element from a set (iii) Running a randomized algorithm and taking the result.}
Equivalently, the influence is $n$ times the probability that a random edge is influential.
Since there are $n \cdot 2^{n-1}$ edges, then $|E_I|=2^{n-1} I[f]$.
\end{definition}

\begin{definition}
Let $\delta$ be a parameter and let $x$ be selected uniformly at random from $\{0,1\}^n$. Let $z \in \{0,1\}^n$ be defined as follows:

\[
z_i=
\begin{cases}
    x_i & \text{with probability } 1-\delta\\
    1-x_i & \text{with probability } \delta
\end{cases}
\]
We denote this distribution of $x$ and $z$ by $T_{\delta}$. Then we define the \textbf{noise sensitivity} of $f$ as:

$$
NS_{\delta}[f] \myeq
Pr_{(x,z) \in_R T_{\delta}}
[f(x) \neq f(z)]
$$
\end{definition}

\begin{observation}
\label{pairing}
For every pair $x^0$ and $z^0$, the probability that for a pair $x,z$ drawn from $T_\delta$, it is the case that $(x,z)=(x^0, z^0)$, is equal to the probability that $(x,z)=(z^0, x^0)$.
Therefore, \[Pr[f(x)=0 \land f(z)=1)=Pr[f(x)=1 \land f(z)=0]\]
Hence: 
\[
NS_{\delta}[f]
=2 \cdot Pr[f(x)=1 \land f(z)=0]
\]
\end{observation}

\subsubsection{Influence estimation.}
To estimate the influence,
standard sampling would require $O\left(\frac{n}{I[f] \epsilon^2 }\right)$ samples. 
However, from \cite{num21} we have:
\begin{theorem}
\label{estimate influence}
There is an algorithm that approximates $I[f]$ to within a multiplicative factor of $(1 \pm \epsilon)$ for a monotone $f$\blue{$:\{0,1\}^n \rightarrow \{0,1\}$}. The algorithm requires that $I[f] \geq 1/n^{C'}$ for a constant $C'$ that is given to the algorithm. It outputs a good approximation with probability at least $0.99$ and in expectation requires $O\left(\frac{\sqrt{n} \log(n/\epsilon)}{I[f] \epsilon^3}\right)$ queries. Additionally, it runs in time polynomial in $n$.
\end{theorem}

%

\subsubsection{Bounds for $\epsilon$ and $I[f]$}

The following observation allows us to assume that without loss of generality $\epsilon$ is not too small. \blue{A similar technique was also used in \cite{num21}.}
\begin{observation}
\label{epsilon is large}
In this work (recalling that $\delta \geq 1/n$), we assume that 
$\epsilon \geq H \sqrt{n}  \delta \log^{1.5}(n) \geq  
H \frac{\log^{1.5} n}{\sqrt{n}}$, \blue{for any constant $H$}, which is without loss of generality for the following reason.
The standard sampling algorithm can estimate the noise sensitivity
in 
$O \left(\frac{1}{NS_\delta[f] \epsilon^2} \right)$ samples.
For 
$\epsilon =O (\sqrt{n}  \delta \log^{1.5}(n))$, this is 
$O \left(\frac{\sqrt{n}  \delta \log^{1.5}(n)}{NS_\delta[f] \epsilon^3}  \right)$ and already accomplishes the desired query complexity.
Additionally, throughout the paper whenever we need it, we will without loss of generality assume that $\epsilon$ is smaller than a sufficiently small positive constant.
\end{observation}
We will also need a known lower bound on influence:
\begin{observation}
\label{I is large}
For any function $f$\blue{$:\{0,1\}^n \rightarrow \{0,1\}$} and $\delta \leq 1/2$ it is the case that:

$$
NS_{\delta}[f]
\leq
\delta I[f]
$$
Therefore it is the case that $I[f] \geq \frac{1}{n^{C}}$.
\end{observation}
A very similar statement is proved in \cite{num9} and for completeness we prove it in Appendix A.

\section{An improved algorithm for small $\delta$}
\label{NS algo}

In this section we give an 
improved algorithm for small $\delta$, namely  
$1/n \leq \delta \leq 1/(\sqrt{n}\log{n})$.
We begin by describing the descending-ascending view on which the algorithm is
based.  

\subsection{Descending-ascending framework}
\label{subsec:view}
\subsubsection{The descending-ascending process.}
It will also be useful to view noise sensitivity in the context of the following process
\footnote{We will use the following notation for probabilities of various events: for an algorithm or a random process $X$ we will use the expression $Pr_{X}[]$ to refer to the random variables we defined in the context of this process or algorithm. It will often be the case that the same symbol, say $x$, will refer to different random random variables in the context of different random processes, so $Pr_{X}[x=1]$ might not be the same as $Pr_{Y}[x=1]$.}.

\hfill \newline
\noindent\fbox{
    \parbox{\textwidth}{
\textbf{Process $D$}
\setlist{nolistsep}
\begin{itemize}[noitemsep]
\item Pick $x$ uniformly at random from $\{0,1\}^n$. Let $S_0$ be the set of indexes $i$ for which $x_i=0$, and conversely let $S_1$ be the rest of indexes.
\item \textbf{Phase 1:} go through all the indexes in $S_1$ in a random order, and flip each with probability $\delta$. Form the descending path $P_1$ from all the intermediate results. Call the endpoint $y$.
\item \textbf{Phase 2:} start at $y$, and flip each index in $S_0$ with probability $\delta$. As before, all the intermediate results form an ascending path $P_2$, which ends in $z$.
\end{itemize}
}}
\hfill \newline

By inspection, $x$ and $z$ are distributed the identically in $D$ as in $T_{\delta}$. Therefore from Observation \ref{pairing}:

$$
NS_{\delta}[f]=2 \cdot Pr_D[f(x)=1 \land f(z)=0]
$$

\begin{observation}
\label{observation xyz}
Since the function is monotone, if $f(x)=1$ and $f(z)=0$, then it has to be that $f(y)=0$.
\end{observation}

Now we define the probability that a Phase 1 path starting somewhere at level $l$ makes at least $w$ steps downwards:
\begin{definition}
For any $l$ and $w$ in $[n]$ we define $Q_{l,w}$ as follows: 
$$
Q_{l;w}
\myeq
Pr_{D}[\bigl\lvert P_1 \bigr\rvert\geq w |L(x)=l]
$$
\end{definition}
This notation is useful when one wants to talk about the probability that a path starting on a particular vertex hits a specific level.

\subsubsection{Defining bad events}
In this section, we give the parameters that we use to determine the lengths
of our walks, as well as the ``middle'' of the hypercube.

Define the following values:
\begin{align*}
t_1 \myeq \eta_1\sqrt{n \log n} 
&&
t_2 \myeq n \delta (1+3\eta_2 \log n)
\end{align*}
Here $\eta_1$ and $\eta_2$ are large enough constants. For constant $C$, taking $\eta_1=\sqrt{C}+4$ and $\eta_2=C+2$ is sufficient for our purposes. 

Informally, $t_1$ and $t_2$ have the following intuitive meaning. A typical vertex $x$ of the hypercube has $L(x)$ between $n/2-t_1$ and $n/2+t_1$. A typical Phase 1 path from process $D$ will have length at most $t_2$. 

We define the ``middle edges" $M$ as the following set of edges:
$$
M \myeq \{e=(v_1,v_2):
\frac{n}{2}-
t_1
 \leq L(v_2)\leq L(v_1) \leq
\frac{n}{2}+
t_1
  \}
$$
Denote by $\overline{M}$ the rest of the edges.

We define the following two bad events in context of $D$. 
The first one happens roughly when $P_1$ (from $x$ to $y$, as defined by
Process D) is much longer than it should be in expectation, and the second one happens when $P_1$ crosses one of the edges that are too far from the middle of the hypercube, which could happen because $P_1$ is long or because of
a starting point that is far from the middle.
\begin{itemize}
\item $E_1$ happens when both of the following hold (i) $P_1$ crosses an edge $e \in E_I$  and (ii) denoting $e=(v_1,v_2)$, so that $v_2 \preceq v_1$, it is the case that  $L(x)-L(v_1) \geq t_2$.
\item $E_2$ happens when $P_1$ contains an edge in $E_I \cap \overline{M}$.
\end{itemize}
While defining $E_1$ we wanted two things from it. First of all, we wanted its probability to be upper-bounded easily. Secondly, we wanted it not to complicate the sampling of paths in Lemma \ref{length sampler}. There exists a tension between these two requirements, and as a result the definition of $E_1$ is somewhat contrived.

\subsubsection{Defining $p_A,p_B$}
\label{pa,pb}

We define:
\begin{align*}
p_A \myeq Pr_{D}[f(x)=1 \land f(y)=0 \land \overline{E_1} \land \overline{E_2}] &&
p_B\myeq Pr_{D}[f(z)=0 | f(x)=1 \land f(y)=0\land \overline{E_1} \land \overline{E_2}]
\end{align*}
Ignoring the bad events, $P_A$ is the probability that $P_1$ hits an influential edge, and $P_B$ is the probability that given that $P_1$ hits an influential edge $P_2$ does not hit an influential edge. From Observation (\ref{observation xyz}), if and only if these two things happen, it is the case that $f(x)=1$ and $f(z)=0$. From this fact and the laws of conditional probabilities we have:
\begin{equation}
\label{eq 1-(4.5)}
Pr_D[f(x)=1 \land f(z)=0 \land \overline{E_1} \land \overline{E_2}]
=Pr_{D}[f(x)=1 \land f(y)=0 \land f(z)=0 \land \overline{E_1} \land \overline{E_2}]=p_A p_B
\end{equation}

We can consider for every individual edge $e$ in $M \cap E_I$ the probabilities:  
$$
p_{e}
\myeq
Pr_{D}[e \in P_1 \land \overline{E_1} \land \overline{E_2}]
$$
$$
q_{e}
\myeq
Pr_{D}[f(x)=1 \land f(z)=0 |e \in P_1 \land \overline{E_1} \land \overline{E_2}]
=
Pr_{D}[f(z)=0 |e \in P_1 \land \overline{E_1} \land \overline{E_2}]
$$
The last equality is true because $e \in P_1$ already implies $f(x)=1$. Informally and ignoring the bad events again, $p_e$ is the probability that $f(x)=1$ and $f(y)=0$ \textbf{because} $P_1$ hits $e$ and not some other influential edge. Similarly, $q_e$ is the probability $f(x)=1$ and $f(z)=0$ given that $P_1$ hits specifically $e$.

Since $f$ is monotone, $P_1$ can hit at most one influential edge. Therefore, the events of $P_1$ hitting different influential edges are disjoint. Using this, Equation (\ref{eq 1-(4.5)}) and the laws of conditional probabilities we can write:

\begin{equation}
\label{eq 1-(6)}
p_A=
\sum_{e \in E_I \cap M}
p_e
\end{equation}
Furthermore, the events that $P_1$ hits a given influential edge and then $P_2$ does not hit any are also disjoint for different influential edges. Therefore, analogous to the previous equation we can write:
\begin{equation}
\label{eq 1-(5)}
p_A p_B=Pr_D[(f(x)=1) \land (f(z)=0) \land \overline{E_1} \land \overline{E_2}]
=\sum_{e \in E_I \cap M}
p_{e} q_{e}
\end{equation}

\subsubsection{Bad events can be ``ignored''}

In the following proof, we will need to consider probability distributions
in which bad events do not happen.  For the most part, conditioning on the
fact that bad events do not happen changes little in the calculations.  In
this subsection, we prove lemmas that allow us to formalize these claims.

The following lemma suggests that almost all influential edges are in $M$.
\begin{observation}
\label{M is large}
It is the case that:
$$
\left(1-\frac{\epsilon}{310}\right)|E_I|\leq|M \cap E_I| \leq
|E_I|
$$

\end{observation}

\begin{proof}
This is the case, because:
$$|\overline{M} \cap E_I| \leq |\overline{M}| \leq 2^n n\cdot 2\exp(-2 \eta_1^2 \log(n))
=2^{n-1} \cdot 4/n^{2\eta_1^2-1}\leq 2^{n-1} I[f]/n 
=|E_I|/n \leq \frac{\epsilon}{310} |E_I|$$

The second inequality is the Hoeffding bound, 
then we used Observations \ref{I is large} and \ref{epsilon is large}.
\end{proof}

\begin{lemma}
\label{corner cutting}
We proceed to prove that ignoring these bad events does not distort our estimate for $NS_{\delta}[f]$. 
It is the case that:

\begin{equation*}
p_A p_B
\leq
\frac{1}{2}NS_{\delta}[f]
\leq
\left(1+\frac{\epsilon}{5} \right)
p_A p_B
\end{equation*}
\end{lemma}
\begin{proof}
See Appendix C, Subsection \ref{proof 2}.
\end{proof}

\subsection{Main lemmas}
\blue{
Here we prove the correctness and run-time of the main subroutines used in our algorithm for estimating noise sensitivity. For completeness, we will repeat all the algorithms.}
\subsection{Lemmas about descending paths hitting influential edges}
Here we prove two lemmas that allow the estimation of the probability that a certain descending random walk hits an influential edge. 
As we mentioned in the introduction, except for the binary search step, 
the algorithm in Lemma \ref{lemma edge sampler} is similar to the algorithm in \cite{num21}. 
In principle, we could have carried out much of the analysis of the algorithm in Lemma \ref{lemma edge sampler} by referencing an equation in \cite{num21}.
However, for subsequent lemmas, 
including Lemma \ref{lemma edges uniform}, 
we build on the application of the Hypercube Continuity Lemma 
to the analysis of random walks on the hypercube.
Thus, we give a full analysis of the algorithm in 
Lemma \ref{lemma edge sampler} here, in order to demonstrate how 
the Hypercube Continuity Lemma (Lemma \ref{continuity}) can be used to analyze random walks on the hypercube, before handling the more complicated
subsequent lemmas, including Lemma \ref{lemma edges uniform}. 
\begin{lemma}
\label{lemma edge sampler}
There exists an algorithm $\mathcal{A}$ that samples edges from $M \cap E_I$ so that for every two edges $e_1$ and $e_2$ in $M \cap E_I$:
$$
\left(1-\frac{\epsilon}{70}\right) Pr_{e \in_R \mathcal{A}}[e=e_2]
\leq
Pr_{e \in_R \mathcal{A}}[e=e_1] \leq \left(1+\frac{\epsilon}{70}\right) Pr_{e \in_R \mathcal{A}}[e=e_2]
$$
The probability that the algorithm succeeds is at least $\frac{1}{O(\sqrt{n } \log^{1.5} n/I[f] \epsilon)}$. If it succeeds, the algorithm makes $O(\log n)$ queries, and if it fails, it makes only $O(1)$ queries. In either case, it runs in time polynomial in $n$.
\end{lemma}

\textbf{Remark:}
Through the standard repetition technique, the probability of error can be decreased to an arbitrarily small constant, at the cost of $O(\frac{ \sqrt{n } \log^{1.5} n}{I[f] \epsilon})$ queries. Then, the run-time still stays polynomial in $n$, since $I[n] \geq 1/n^C$.

\textbf{Remark:}
The distribution $\mathcal{A}$ outputs is point-wise close to the uniform distribution over $M \cap E_I$. We will also obtain such approximations to other distributions in further lemmas. Note that this requirement is stronger than closeness in $L_1$ norm. 
 
\begin{proof}
The algorithm is as follows:

\hfill \newline
\noindent\fbox{
    \parbox{\textwidth}{
\textbf{Algorithm $\mathcal{A}$} (given oracle access to a monotone function $f:\{0,1\}^n \rightarrow \{0,1\}$ and a parameter $\epsilon$)
\setlist{nolistsep}
\begin{enumerate}[noitemsep]
\item Assign $w=\frac{\epsilon}{3100 \eta_1} \sqrt{\frac{n}{\log n}}$ 
\item Pick $x$ uniformly at random from $\{0,1\}^n$.
\item Perform a descending walk $P_1$ downwards in the hypercube starting at $x$. Stop at a vertex $y$ either after $w$ steps, or if you hit the all-zeros vertex. 
\item If $f(x)=f(y)$ output FAIL.
\item If $f(x) \neq f(y)$ perform a binary search on the path $P_1$ and find an influential edge $e_{inf}$.
\item If $e_{inf} \in M$ return $e_{inf}$. Otherwise output FAIL.
\end{enumerate}
}}
\hfill \newline

Note that $P_1$ is distributed symmetrically with respect to change of indices. Pick an arbitrary edge $e_0$ in $M \cap E_{I}$ and let $e_0=(v_1, v_2)$, so $v_2 \preceq v_1$. Recall that $\Lambda_{e_0}$ was the set of edges in the same ``layer" of the hypercube and that $|\Lambda_{e_0}|$ equals $L(v_1){{n}\choose{L(v_1)}}$. 

For any $e'$ in $\Lambda_{e_0}$, $e_0$ and $e'$ are different only up to a permutation of indexes. Therefore, the probability that $P_1$ passes through $e'$ equals to the probability of it passing through $e_0$. Additionally, since $P_1$ is descending it can cross at most one edge in $\Lambda_{e_0}$, which implies that the events of $P_1$ crossing each of these edges are disjoint. Thus, we can write:

$$
Pr_{\mathcal{A}}
[e_0 \in P_1]
=
\frac{Pr_{\mathcal{A}}
[P_1 \cap \Lambda_{e_0} \neq \emptyset]}{L(v_1){{n}\choose{L(v_1)}}}
$$
But $P_1$ will intersect $\Lambda_{e_0}$ if and only if $L(v_1)\leq L(x) \leq L(v_1)+w-1$. This allows us to express the probability in the numerator as a sum over the $w$ layers of the hypercube right above $\Lambda_{e_0}$.
\begin{equation}
\label{eq 1-(7)}
Pr_{\mathcal{A}}
[e_0 \in P_1]
=
\frac{\sum_{l=L(v_1)}^{L(v_1)+w-1} \frac{1}{2^n} {{n}\choose{l}}
}{L(v_1){{n}\choose{L(v_1)}}}
\end{equation}
Observation \ref{epsilon is large} allows us to lower-bound $\epsilon$ and argue that for sufficiently large $n$ we have:
\begin{equation}
\label{eq 1-(8)}
\left(1-\frac{\epsilon}{1300}\right)\frac{2}{n}
\leq
\frac{2}{n}
\frac{1}{1+2t_1/n}
\leq
\frac{1}{L(v_1)}
\leq
\frac{2}{n} \frac{1}{1-2t_1/n}
\leq
\left(1+\frac{\epsilon}{1300}\right)\frac{2}{n}
\end{equation}

Since $e \in M$, it is the case that $\frac{n}{2}-
t_1
 \leq L(v_2) \leq L(v_1) \leq
\frac{n}{2}+
t_1$. This allows us to use Lemma \ref{continuity} and deduce that $1-\epsilon/310
\leq
{{n}\choose{l}}/{{{n}\choose{L(v_1)}}}
\leq
1+\epsilon/310
$. This and Equation (\ref{eq 1-(8)}) allow us to approximate $Pr_{\mathcal{A}}[e_0 \in P_1]$ in Equation (\ref{eq 1-(7)}) the following way:

\begin{equation}
\label{eq 1-(8.1)}
\left(1-\frac{\epsilon}{150}\right)\frac{w}{n2^{n-1}}
\leq
Pr_{\mathcal{A}}
[e_0 \in P_1]
\leq \left(1+\frac{\epsilon}{150}\right)\frac{w}{n2^{n-1}}
\end{equation}
The algorithm outputs an influential edge if and only if $P_1$ hits an influential edge in $M \cap E_I$. At the same time, these events corresponding to different edges in $M \cap E_I$ are disjoint since $P_1$ can hit only at most one influential edge. This together with Bayes rule allows us to express the probability that the algorithm outputs $e_0$, conditioned on it succeeding:
\begin{equation}
\label{eq 1-(8.2)}
Pr_{e \in_R \mathcal{A}}[e=e_0]
=
Pr_{\mathcal{A}}
[e_{inf}=e_0]
=
Pr_{\mathcal{A}}
\left[e_0 \in P_1 \bigg \vert \bigvee_{e' \in M \cap E_I} (e' \in P_1) \right]
=
\frac{Pr_{\mathcal{A}}
\left[e_0 \in P_1 \right]}{Pr_{\mathcal{A}}
\left[\bigvee_{e' \in M \cap E_I} (e' \in P_1) \right]}
\end{equation}
Substituting Equation (\ref{eq 1-(8.2)}) into Equation (\ref{eq 1-(8.1)}) we get:

\begin{equation*}
\left(1-\frac{\epsilon}{150}\right)\frac{w}{n2^{n-1}}
\cdot
Pr_{\mathcal{A}}
\left[\bigvee_{e_1 \in M \cap E_I} (e_1 \in P_1) \right]
\leq
Pr_{e \in_R \mathcal{A}}[e=e_0]
\leq \left(1+\frac{\epsilon}{150}\right)\frac{w}{n2^{n-1}}
\cdot
Pr_{\mathcal{A}}
\left[\bigvee_{e_1 \in M \cap E_I} (e_1 \in P_1) \right]
\end{equation*}

Substituting two different edges $e_1$ and $e_2$ in $E_I \cap M$ in place of $e_0$ and then dividing the resulting inequalities gives us:
$$
\frac{1-\epsilon/150}{1+\epsilon/150}
Pr_{e \in_R \mathcal{A}}[e=e_2]
\leq
Pr_{e \in_R \mathcal{A}}[e=e_1] \leq 
\frac{1+\epsilon/150}{1-\epsilon/150} Pr_{e \in_R \mathcal{A}}[e=e_2]
$$
From this, the correctness of the algorithm follows. In case of failure, it makes $O(1)$ queries and in case of success, it makes $O(\log \epsilon +\log(n/\log n))=O(\log n)$ queries, because of the additional binary search. In either case, the run-time is polynomial. 

Now, regarding the probability of success, note that the events of $P_1$ crossing different edges in $E_I \cap M$ are disjoint since the function is monotone. Therefore, we can sum Equation (\ref{eq 1-(8.1)}) over all edges in $M \cap E_I$ and get that the probability of success is at least $\Theta \left(\frac{|E_I \cap M| w}{n 2^{n-1}} \right)$. Applying Lemma \ref{M is large} and substituting $w$, the inverse of the success probability is:
$$
O \left(\frac{n 2^{n-1}}{|E_I \cap M| w} \cdot \log n \right)
=
O \left(\frac{n}{I[f] w} \cdot \log n \right)
=
O \left(\frac{\log^{1.5}(n) \sqrt{n }}{I[f]\epsilon} \right)
$$ 
\end{proof}

The following lemma, roughly speaking, shows that just as in previous lemma, the probability that $P_1$ in $D$ hits an influential edge $e$ does not depend on where exactly $e$ is, as long as it is in $M \cap E_I$. The techniques we use are similar to the ones in the previous lemma and it follows the same outline. However here we encounter additional difficulties for two reasons: first of all, the length of $P_1$ is not fixed, but it is drawn from a probability distribution. Secondly, this probability distribution depends on the starting point of $P_1$. 

Note that unlike what we have in the previous lemma, here $P_1$ comes from the ascending-descending view of noise sensitivity.
\begin{lemma}
\label{lemma edges uniform}
For any edge $e \in M \cap E_I$ it is the case that:

\begin{equation*}
\left(1-\frac{\epsilon}{310}\right)\frac{\delta}{2^{n}}
\leq
p_{e}
\leq
\left(1+\frac{\epsilon}{310}\right)\frac{\delta}{2^{n}}
\end{equation*}
\end{lemma}
\begin{proof}
Let $e=(v_1, v_2)$, so $v_2 \preceq v_1$. 
We can use the same argument from symmetry as in the proof of Lemma \ref{lemma edge sampler}. 
We get:
\begin{multline*}
p_{e}
=
Pr_{D}[(e \in P_1) \land \overline{E_1} \land \overline{E_2}]=
Pr_{D}[(e \in P_1) \land \overline{E_1}]
\\=
Pr_{D}[(e \in P_1) \land (L(x)-L(v_1) < t_2) ]
=
\frac{Pr_{D}
[(P_1 \cap \Lambda_e\neq \emptyset)\land (L(x)-L(v_1) < t_2)]}{L(v_1){{n}\choose{L(v_1)}}}
\end{multline*}
Above we did the following: Recall that $E_2$ is the event that $P_1$ crosses an edge in $E_I \cap \overline{M}$. The first equality is true because $e\in E_I \cap M$ and $P_1$ can cross at most one influential edge, which implies that $E_2$ cannot happen. For the second equality we substituted the definition of $E_1$. For the third equality we used the symmetry of $D$ with respect to change of indexes just as in the proof of Lemma \ref{lemma edge sampler}.

Recall that we defined:
$$
Q_{l;w}
\myeq
Pr_{D}[\bigl\lvert P_1 \bigr\rvert\geq w |L(x)=l]
$$

This allows us to rewrite:

\begin{equation}
\label{eq 1-(9)}
p_{e}
=
\frac{\sum_{i=1}^{t_2} \frac{1}{2^n} {{n}\choose{L(v_1)+i-1}} Q_{L(v_1)+i-1;i}
}{L(v_1){{n}\choose{L(v_1)}}}
\end{equation}
Above we just looked at each of the layers of the hypercube and summed the contributions from them. 

To prove a bound on $p_e$ we will need a bound on $t_2$. Observation \ref{epsilon is large} implies that for any $H$ we can assume that $\sqrt{n} \delta \leq \frac{\epsilon}{H \log^{1.5} n}$ using which we deduce $$t_2=n \delta(1+3 \eta_2 \log(n))\leq 4 \eta_2 n \delta  \log(n) \leq \frac{4\eta_2 \epsilon}{H}\sqrt{\frac{n}{\log n}}$$
Furthermore, we will need a bound on the binomial coefficients in Equation (\ref{eq 1-(9)}). By picking $H$ to be a large enough constant, this allows us to use Lemma \ref{continuity} to bound the ratios of the binomial coefficients. For any $i$ between $1$ and $t_2$ inclusive:
\begin{equation}
\label{eq 1-(10)}
1-\frac{\epsilon}{1300}
\leq
\frac{{{n}\choose{L(v_1)+i-1}}}{{{n}\choose{L(v_1)}}}
\leq
1+\frac{\epsilon}{1300}
\end{equation}
Equation (\ref{eq 1-(8)}) is valid in this setting too. Substituting Equation (\ref{eq 1-(8)}) together with Equation (\ref{eq 1-(10)}) into Equation (\ref{eq 1-(9)}) we get:

\begin{equation}
\label{eq 1-(11)}
\left(1-\frac{\epsilon}{630}\right)
\frac{1}{n 2^{n-1}}
\sum_{i=1}^{t_2} Q_{L(v_1)+i-1;i}
\leq
p_{e}
\leq
\left(1+\frac{\epsilon}{630}\right)
\frac{1}{n 2^{n-1}}
\sum_{i=1}^{t_2} Q_{L(v_1)+i-1;i}
\end{equation}

Observe that a vertex on the level $l+1$ has more ones than a vertex on level $l$. So for a positive integer $w$, the probability that at least $w$ of them will flip is larger. Therefore, $Q_{l,w} \leq Q_{l+1,w}$. This allows us to bound:
\begin{equation}
\label{eq 1-(12)}
\sum_{i=1}^{t_2} Q_{L(v_1);i}
\leq
\sum_{i=1}^{t_2} Q_{L(v_1)+i-1;i}
\leq
\sum_{i=1}^{t_2} Q_{L(v_1)+t_2-1;i}
\end{equation}
Then, using Observation \ref{epsilon is large}:

\begin{multline}
\label{eq 1-(13)}
\sum_{i=1}^{t_2} Q_{L(v_1)+t_2-1;i}
\leq
\sum_{i=1}^{n} Q_{L(v_1)+t_2-1;i}
=
E_{D}[
L(x)-L(y)
|L(x)=L(v_1)+t_2-1]
\\
=
\delta(
L(v_1)+t_2-1
)
\leq
\delta(
n/2+t_1+t_2)
\leq
\left(1+\frac{\epsilon}{630}\right)\frac{n\delta}{2}
\end{multline}

Now, we bound from the other side:

\begin{equation}
\label{eq 1-(14)}
\sum_{i=1}^{t_2} Q_{L(v_1);i}=
\sum_{i=1}^{n} Q_{L(v_1);i}-\sum_{i=t_2+1}^{n} Q_{L(v_1);i}
\geq
\sum_{i=1}^{n} Q_{L(v_1);i}-n \cdot Q_{L(v_1);t_2}
\end{equation}
We bound the first term just as in Equation (\ref{eq 1-(13)}), and we bound the second one using a Chernoff bound:
\begin{equation}
\label{eq 1-(15)}
\sum_{i=1}^{n} Q_{L(v_1);i}
=
\delta L(v_1)
\geq
\delta(n/2-t_1)
\geq
\left(1-\frac{\epsilon}{1300}\right)\frac{\delta n}{2}
\end{equation}
\begin{equation}
\label{eq 1-(16)}
n \cdot Q_{L(v_1);t_2}
\leq
n \cdot \exp(
-
\frac{1}{3}
n \cdot \delta 3 
\eta_2 \log n
)
\leq
\frac{1}{n^{\eta_2-1}}
\end{equation}
Substituting Equations (\ref{eq 1-(15)}) and (\ref{eq 1-(16)}) into Equation (\ref{eq 1-(14)}) and using Observation \ref{epsilon is large} we get:
\begin{equation}
\label{eq 1-(17)}
\sum_{i=1}^{t_2} Q_{L(v_1);i}
\geq
\left(1-\frac{\epsilon}{1300}\right)\frac{\delta n}{2}
-\frac{1}{n^{\eta_2-1}}
\geq
\left(1-\frac{\epsilon}{630}\right)\frac{\delta n}{2}
\end{equation}
Substituting Equations (\ref{eq 1-(17)}) and (\ref{eq 1-(13)}) into Equation (\ref{eq 1-(12)}) we get:
\begin{equation*}
\left(1-\frac{\epsilon}{630}\right)\frac{\delta n}{2}
\leq
\sum_{i=1}^{t_2} Q_{L(v_1)+i-1;i}
\leq
\left(1+\frac{\epsilon}{630}\right)\frac{\delta n}{2}
\end{equation*}
Combining this with Equation (\ref{eq 1-(11)}) we deduce that:
$$
\left(1-\frac{\epsilon}{310}\right)\frac{\delta}{2^n}
\leq
p_{e}
\leq
\left(1+\frac{\epsilon}{310}\right)\frac{\delta}{2^n}
$$
\end{proof}
\subsection{Sampling descending and ascending paths going through a given influential edge.}

While we will use Lemma \ref{lemma edges uniform} in order to estimate $p_A$, we will use the machinery developed in this section to estimate $p_B$ in Section \ref{the ns estimation algo}. To that end, we will need to sample from a distribution of descending and ascending paths going through a given edge. The requirement on the distribution is that it should be close to the conditional distribution of such paths $P_1$ that would arise from process $D$, conditioned on going through $e$ and satisfying $\bar{E}_1$ and $\bar{E}_2$. See a more formal 
explanation in the statements of the lemmas.

In terms of resource consumption, the algorithms in this section require no queries to $f$ but only run-time. Note that the following simple exponential time algorithm achieves the correctness guarantee of Lemma \ref{lemma 3} and still does not need to make any new queries to the function: The algorithm repeatedly samples $P_1$ and $P_2$ from the process $D$ until it is the case that $P_1$ crosses the given edge $e$ and neither $E_1$ nor $E_2$ happen. When this condition is satisfied the algorithm outputs these paths $P_1$ and $P_2$. The resulting distribution would exactly equal the distribution we are trying to approximate in Lemma \ref{lemma 3}, which is the ultimate goal of the section. The polynomial run-time constraint compels us to do something else. Furthermore, this is the only part of the algorithm for which the polynomial time constraint is non-trivial.

A first approach to sampling $P_1$ would be to take random walks in opposite directions from the endpoints of the edge $e$ and then concatenate them together. This is in fact what we do. However, difficulty comes from determining the appropriate lengths of the walks for the following reason. If $P_1$ is longer, it is more likely to hit the influential edge $e$. This biases the distribution of the descending paths hitting $e$ towards the longer descending paths. In order to accommodate for this fact we used the following two-step approach: 
\begin{enumerate}
\item Sample only the levels of the starting and ending points of the path $P_1$. This is equivalent to sampling the length of the segment of $P_1$ before the edge $e$ and after it. This requires careful use of rejection sampling together with the techniques we used to prove Lemmas \ref{lemma edge sampler} and \ref{lemma edges uniform}. Roughly speaking, we use the fact that $P_1$ is distributed symmetrically with respect to the change of indices in order to reduce a question about the edge $e$ to a question about the layer $\Lambda_e$. Then, we use the Lemma \ref{continuity} to answer questions about random walks hitting a given layer. This is handled in Lemma \ref{length sampler}.
\item Sample a path $P_1$ that has the given starting and ending levels and passes through an influential edge $e$. This part is relatively straightforward. We prove that all the paths satisfying these criteria are equally likely. We sample one of them randomly by performing two random walks in opposite directions starting at the endpoints of $e$. This all is handled in Lemma \ref{lemma 3}.
\end{enumerate}

\begin{lemma}
\label{length sampler}
There is an algorithm $\mathcal{W}$ that takes as input an edge $e=(v_1, v_2)$ in  $M \cap E_I$, so that $v_2 \preceq v_1$, and samples two non-negative numbers $w_1$ and $w_2$, so that for any two non-negative $w_1'$ and $w_2'$:
\begin{multline}
\label{eq 1-(18)}
\left(1-\frac{\epsilon}{70}\right)
Pr_{\mathcal{W}(e)}
[(w_1=w_1') \land (w_2=w_2')
]
\\ \leq
Pr_{D}
[
(
L(x)-L(v_1)
=w_1')
\land
(L(v_2)-L(y)
=w_2')
\vert
(e \in P_1)
\land \overline{E_1}
\land \overline{E_2}
]
\\
\leq
\left(1+\frac{\epsilon}{70}\right)
Pr_{\mathcal{W}(e)}
[(w_1=w_1') \land (w_2=w_2')
]
\end{multline}
The algorithm requires no queries to $f$ and runs in time polynomial in $n$.
\end{lemma}
\textbf{Remark:} the approximation guarantee here is similar to the one in Lemma \ref{lemma edge sampler}. It guaranteed that the relative distance should be small point-wise. This guarantee is stronger than closeness in either $L_1$ and $L_\infty$ norms. We also employ an analogous approximation guarantee in Lemma \ref{lemma 3}.

\begin{proof}
Here is the algorithm (recall that we used the symbol $\Lambda_e$ to denote the whole layer of edges that have the same endpoint levels as $e$):

\hfill \newline
\noindent\fbox{
    \parbox{\textwidth}{
\textbf{Algorithm $\mathcal{W}$} (given an edge $e\myeq(v_1, v_2)$ so $v_1 \preceq v_2$)
\setlist{nolistsep}
\begin{enumerate}[noitemsep]
\item Pick an integer $l$ uniformly at random among the integers in $[L(v_1) ,L(v_1)+t_2-1]$. Pick a vertex $x$ randomly at level $l$.
\item As in phase 1 of the noise sensitivity process, traverse in random order through the indices of $x$ and with probability $\delta$, flip each index that equals to one to zero. The intermediate results form a path $P_1$, and we call its endpoint $y$.
\item If $P_1$ does not intersect $\Lambda_e$ go to step $1$.
\item Otherwise, output $w_1=L(x)-L(v_1)$ and $w_2=L(v_2)-L(y)$.
\end{enumerate}
}}
\hfill \newline
To prove the correctness of the algorithm, we begin by simplifying the expression above. If $e \in P_1$, then $P_1$ cannot contain any other influential edge, so $E_2$ cannot happen and we can drop it from notation. Additionally we can substitute the definition for $E_1$ so:
\begin{multline*}
Pr_{D}
[
(
L(x)-L(v_1)
=w_1')
\land
(L(v_2)-L(y)
=w_2')
\vert
e \in P_1
\land \overline{E_1}
\land \overline{E_2}
]
\\
=
Pr_{D}
\left[
(
L(x)-L(v_1)
=w_1')
\land
(
L(v_2)-L(y)
=w_2')
\bigg \vert
e \in P_1
\land (L(x)-L(v_1) < t_2)
\right]
\end{multline*}

Now we can use the fact that $D$ is symmetric with respect to the change of indices, so we can substitute $e$ with any $e'$ in $\Lambda_e$. Therefore:

\begin{multline}
\label{eq 1-(19)}
Pr_{D}
[
(
L(x)-L(v_1)
=w_1')
\land
(L(v_2)-L(y)
=w_2')
\vert
(
e \in P_1)
\land \overline{E_1}
\land \overline{E_2}
]
\\
=
\sum_{e' \in \Lambda_e}
Pr_{D}
\left[
(L(x)-L(v_1)
=w_1')
\land
(L(v_2)-L(y)
=w_2')
\bigg \vert
e' \in P_1
\land (L(x)-L(v_1) < t_2)
\right]
\cdot
\frac{1}{|\Lambda_e|}
\\
=
\sum_{e' \in \Lambda_e}
Pr_{D}
\left[
L(x)-L(v_1)
=w_1'
\land
(L(v_2)-L(y)
=w_2')
\bigg \vert
e' \in P_1
\land (L(x)-L(v_1) < t_2)
\right]
\\
\times
Pr_{D}\left[e' \in P_1 \bigg \vert P_1 \cap \Lambda_e \neq \emptyset\right]
\\ =
Pr_{D}
\left[(
L(x)-L(v_1)
=w_1')
\land
(L(v_2)-L(y)
=w_2')
\bigg\vert
(P_1 \cap \Lambda_e \neq \emptyset)
\land (L(x)-L(v_1) < t_2)
\right]
\end{multline}

\begin{observation}
\label{obs w1 w2}
The algorithm never returns $w_1 \geq t_2$, which is appropriate since the distribution being approximated is conditioned on $L(x)-L(v_1) < t_2$. In what follows we consider $1\leq w_1 < t_2$ and $1\leq w_1' < t_2$. Additionally, recall that $w_2'$ is non-negative.
\end{observation}

Before we continue, we derive the following intermediate results.
Because of how $P_1$ is chosen in $\mathcal{W}$, we can use the notation:
$$
Q_{l;w}
=
Pr_{D}[\bigl\lvert P_1 \bigr\rvert\geq w |L(x)=l]
=
Pr_{\mathcal{W}(e)}[\bigl\lvert P_1 \bigr\rvert\geq w |L(x)=l]
$$
The last equality is true because after picking $x$ both process $D$ and algorithm $\mathcal{W}(e)$ pick the path $P_1$ exactly the same way. Therefore, the path $P_1$ is distributed the same way for the both processes if we condition on them picking the same starting point $x$. 

We can express:

\begin{multline}
\label{eq 1-(19.5)}
Pr_{\mathcal{W}(e)}
[
(L(x)-L(v_1)=w_1')
\land
(L(v_2)-L(y)=w_2')
]
=
\frac{1}{t_2}
Pr_{\mathcal{W}(e)}
[
L(v_2)-L(y)=w_2'
\vert
L(x)-L(v_1)=w_1'
]
\\=
\frac{1}{t_2}
Pr_{D}
[
L(v_2)-L(y)=w_2'
\vert
L(x)-L(v_1)=w_1'
]
\\=
\frac{1}{t_2}Q_{L(v_1)+w_1';w_1'+1}
Pr_{D}\left[
L(v_2)-L(y)=w_2' \bigg \vert (L(x)-L(v_1)= w_1')
\land (|P_1| \geq w_1'+1)
\right]
\end{multline}
Above: (i) The first equality uses the fact that in $\mathcal{W}(e)$ the level of $x$ is chosen uniformly between the $t_2$ levels, and therefore  $Pr_{\mathcal{W}(e)}[L(x)-L(v_1)=w_1']=1/t_2$. (ii) The second inequality uses the fact that after picking $x$, both $W(e)$ and $D$ pick $P_1$ and consequently $y$ identically. (iii) The third inequality is an application of the law of conditional probabilities together with the definition of $Q_{l,w}$. 

Building on the previous equality we have:

\begin{multline}
\label{eq 1-(20)}
Pr_{\mathcal{W}(e)}
[
(w_1=w_1')
\land
(w_2=w_2')
]
=
Pr_{\mathcal{W}(e)}
\left[
(L(x)-L(v_1)=w_1')
\land
(L(v_2)-L(y)=w_2')
\bigg \vert
P_1 \cap \Lambda_e \neq \emptyset
\right]
\\
=
\frac{Pr_{\mathcal{W}(e)}
[
(L(x)-L(v_1)=w_1')
\land
(L(v_2)-L(y)=w_2')
\land (P_1 \cap \Lambda_e \neq \emptyset)
]}{Pr_{\mathcal{W}}[P_1 \cap \Lambda_e \neq \emptyset
]]}
\\=
\frac{Pr_{\mathcal{W}(e)}
[
(L(x)-L(v_1)=w_1')
\land
(L(v_2)-L(y)=w_2')
]}{Pr_{\mathcal{W}}[P_1 \cap \Lambda_e \neq \emptyset
]]}
\\=
\frac{Q_{L(v_1)+w_1';w_1'+1}
Pr_{D}\left[
L(v_2)-L(y)=w_2' \bigg \vert (L(x)=L(v_1)+w_1')\land (|P_1| \geq w_1'+1)\right]}
{\sum_{i=1}^{t_2} Q_{L(v_1)+i-1;i}
}
\end{multline}
Above: (i) At step (5) of the algorithm, we have $w_1=L(x)-L(v_1)$ and $w_2=L(v_2)-L(y)$, but this happens only after the condition on step (4) is satisfied. This adds a conditioning at the first equality. (ii) The second equality comes from Observation \ref{obs w1 w2} since $w_1' \geq 1$ and $w_2'\geq 0$, hence $\Lambda_e$ starts above $\Lambda_e$ and ends below it. Thus, the first two clauses imply the last one, so we drop it. (iii) Regarding the third equality, in the numerator we substituted Equation (\ref{eq 1-(19.5)}) whereas in the denominator we computed $Pr_{\mathcal{W}}[P_1 \cap \Lambda_{e} \neq \emptyset]$ by breaking it into contributions from the $t_2$ levels above $\Lambda_e$. Finally, we canceled $\frac{1}{t_2}$ from both the numerator and the denominator.

Back to Equation (\ref{eq 1-(19)}). Analogous to how we derived Equation (\ref{eq 1-(19.5)}), we have:

\begin{multline*}
Pr_{D}
[(
L(x)-L(v_1)
=w_1')
\land(
L(v_2)-L(y)
=w_2')
]
\\ =
Pr_{D}
[
L(x)-L(v_1)
=w_1'
]
\cdot
Pr_{D}
\left[
L(v_2)-L(y)
=w_2'
\bigg \vert
L(x)-L(v_1)
=w_1'
\right]
\\
=
\frac{1}{2^n} {{n}\choose{L(v_1)+w_1'}} Q_{L(v_1)+w_1';w_1'+1}
\cdot
Pr_{D}\left[
L(v_2)-L(y)=w_2' \bigg \vert (L(x)-L(v_1)=w_1') \land (|P_1| \geq w_1'+1)
\right]
\end{multline*}
In addition to the steps analogous to the ones involved in getting Equation (\ref{eq 1-(19)}), above we used the fact that the layer $L(v_1)+w_1'$ has ${{n}\choose{L(v_1)+w_1'}} Q_{L(v_1)+w_1';w_1'+1}$ vertices out of the $2^n$ vertices overall.

Again, the same way we derived Equation (\ref{eq 1-(20)}) from Equation (\ref{eq 1-(19.5)}) using Bayes rule, we get:

\begin{multline}
\label{eq 1-(21)}
Pr_{D}
\left[
(L(x)-L(v_1)
=w_1')
\land
(L(v_2)-L(y)
=w_2')
\bigg \vert
(P_1 \cap \Lambda_e \neq \emptyset )
\land (L(x)-L(v_1) < t_2)\right]
\\
=
\frac{Pr_{D}
[
(L(x)-L(v_1)
=w_1')
\land
(L(v_2)-L(y)
=w_2')
\land 
(P_1 \cap \Lambda_e \neq \emptyset )
\land (L(x)-L(v_1) < t_2)
]}{Pr_{D}[(P_1 \cap \Lambda_e \neq \emptyset) 
\land (L(x)-L(v_1) < t_2)]}
\\=
\frac{Pr_{D}
[
(L(x)-L(v_1)
=w_1')
\land
(L(v_2)-L(y)
=w_2')
]}{Pr_{D}[(P_1 \cap \Lambda_e \neq \emptyset) 
\land (L(x)-L(v_1) < t_2)]}
\\
=
\frac{\frac{1}{2^n}{{n}\choose{L(v_1)+w_1'}} Q_{L(v_1)+w_1';w_1'+1}
\cdot
Pr_{D}\left[
L(v_2)-L(y)=w_2' \bigg\vert (L(x)-L(v_1)=w_1')\land (|P_1| \geq w_1'+1)\right]}{\sum_{i=1}^{t_2} \frac{1}{2^n} {{n}\choose{L(v_1)+i-1}} Q_{L(v_1)+i-1;i}}
\end{multline}
In the second equality we dropped the clauses $(P_1 \cap \Lambda_e \neq \emptyset)$ and $(L(x)-L(v_1) < t_2)$ because by Observation \ref{obs w1 w2} they are implied by the first two clauses.

Since by Observation \ref{obs w1 w2} it is the case that $1 \leq i \leq t_2$ and $0 \leq w_1' < t_2$, the same way we proved Equation (\ref{eq 1-(10)}) we have:

\begin{equation}
\label{eq 1-(22)}
1-\frac{\epsilon}{150}
\leq
\frac{{{n}\choose{L(v_1)+i-1}}}{{{n}\choose{L(v_1)+w_1'}}}
\leq
1+\frac{\epsilon}{150}
\end{equation}
Combining Equations (\ref{eq 1-(21)}) and (\ref{eq 1-(22)}) we get:
\begin{multline}
\label{eq 1-(23)}
\left(1-\frac{\epsilon}{70}\right)
\frac{Q_{L(v_1)+w_1';w_1'+1}
\cdot
Pr_{D}[
L(v_2)-L(y)=w_2' \vert (L(x)-L(v_1)=w_1')\land (|P_1| \geq w_1'+1)]}{\sum_{i=1}^{t_2} Q_{L(v_1)+i-1;i}}
\\
\leq
Pr_{D}
[(
L(x)-L(v_1)
=w_1')
\land
(L(v_2)-L(y)
=w_2')
\vert
(P_1 \cap \Lambda_e \neq \emptyset)
\land (L(x)-L(v_1) \leq t_2)
]
\\
\leq
\left(1+\frac{\epsilon}{70}\right)
\frac{Q_{L(v_1)+w_1';w_1'+1}
\cdot
Pr_{D}[
L(v_2)-L(y)=w_2' \vert (L(x)-L(v_1)=w_1')\land (|P_1| \geq w_1'+1)]}{\sum_{i=1}^{t_2} Q_{L(v_1)+i-1;i}}
\end{multline}
Substituting Equations (\ref{eq 1-(19)}) and (\ref{eq 1-(20)}) into Equation (\ref{eq 1-(23)}) proves the correctness of the algorithm. 

Now, we consider the run-time. Since $L(x) \in [L(v_1), L(v_1)+t_2-1]$ and $L(v_1) \geq n/2-t_1 \geq n/4$, then $L(x) \geq n/4$. In expectation at least $\delta n/4$ of these indices, which equal to one, should flip. We can use a Chernoff bound to bound the probability of less than $\delta n/8$ of them flipping by $\exp\left(-\frac{1/4 \cdot n \delta/4
}{2}\right)=\exp\left(-\frac{n \delta}{32}\right) \leq \exp(-\frac{1}{32})$. Therefore:

$$
Pr_{\mathcal{W}}[|P_1| \geq n \delta/8]
\geq
\Theta(1)
$$

If this happens, it is sufficient for $l$ to be less than $L(v_1) +n \delta/8$ for $P_1$ to intersect $\Lambda_e$. The probability of this happening is at least $\frac{n \delta/8}{t_2}$. Therefore, we can conclude:

$$
Pr_{\mathcal{W}}[P_1 \cap \Lambda_e \neq \emptyset]
\geq
\Omega\left(\frac{n \delta}{t_2}\right)
$$
Then, the number of time the algorithm goes through the loop is $O(t_2/(n \delta))=\tilde{O}(1)$. Thus, the algorithm runs in polynomial time.
\end{proof}

\begin{lemma}
\label{lemma 3}
There exists an algorithm $\mathcal{B}$ with the following properties. It takes as input an edge $e=(v_1, v_2)$ in $M \cap E_I$, so that $v_2 \preceq v_1$ and outputs paths $P_1$ and $P_2$  together with hypercube vertices $x, y$ and $z$. It is the case that $x$ is the starting vertex of $P_1$, $y$ is both the starting vertex of $P_2$ and the last vertex of $P_1$, and $z$ is the last vertex of $P_2$. Additionally, $P_1$ is descending and $P_2$ is ascending. Furthermore, for any pair of paths $P_1'$ and $P_2'$ we have:

\begin{multline}
\label{eq 1-(24)}
\bigl \lvert
Pr_{\mathcal{B}(e)}[(P_1=P_1') \land (P_2=P_2')]
-
Pr_{D}[(P_1=P_1') \land (P_2=P_2')\vert (e \in P_1) \land \overline{E_1} \land \overline{E_2}]
\bigr \rvert
\\ \leq
\frac{\epsilon}{70}
Pr_{\mathcal{B}(e)}[(P_1=P_1') \land (P_2=P_2')]
\end{multline}

It requires no queries to the function and takes computation time polynomial in $n$ to draw one sample.
\end{lemma}
\begin{proof}
Below we give the algorithm:
\hfill \newline
\noindent\fbox{
    \parbox{\textwidth}{
\setlist{nolistsep}
\textbf{Algorithm $\mathcal{B}$} (given an influential edge $e$)
\begin{enumerate}[noitemsep] 
\item Use $\mathcal{W}(e)$ to sample $w_1$ and $w_2$.
\item Perform an ascending random walk of length $w_1$ starting at $v_1$ and call its endpoint $x$. Similarly, perform a descending random walk starting at $v_2$ of length $w_2$, call its endpoint $y$.
\item Define $P_1$ as the descending path that results between $x$ and $y$ by concatenating the two paths from above, oriented appropriately, and edge $e$. 
\item Define $P_2$ just as in phase 2 of our process starting at $y$. Consider in random order all the zero indices $y$ has in common with $x$ and flip each with probability $\delta$.
\item Return $P_1$ ,$P_2$, $x$, $y$ and $z$. 
\end{enumerate}
}}
\hfill \newline
Now, we analyze the algorithm.
Without loss of generality we assume that $P_1'$ is descending and starts at a vertex $x'$ and ends at a vertex $y'$. $P_2'$, in turn, is ascending, starts at $y'$ and ends at $z'$. For all the other paths all the probabilities in (24) equal to zero. We have that:

\begin{multline}
\label{eq 1-(25)}
Pr_{D}\left[(P_1=P_1') \land (P_2=P_2')
\bigg \vert
(e \in P_1) \land \overline{E_1} \land \overline{E_2}\right]
\\ =
Pr_{D}\left[(L(x)=L(x')) \land (L(y)=L(y'))
\bigg \vert
(e \in P_1) \land \overline{E_1} \land \overline{E_2}\right]
\\
\times
Pr_{D}\left[(P_1=P_1') \land (P_2=P_2')
\bigg \vert 
(e \in P_1) \land \overline{E_1} \land \overline{E_2} \land (L(x)=L(x')) \land (L(y)=L(y'))\right]
\end{multline}
Similarly, we can write:
\begin{multline}
\label{eq 1-(26)}
Pr_{\mathcal{B}(e)}[(P_1=P_1') \land (P_2=P_2')]
=
Pr_{\mathcal{B}(e)}[(L(x)=L(x')) \land (L(y)=L(y'))]
\\
\times
Pr_{\mathcal{B}(e)}
\left[(P_1=P_1') \land (P_2=P_2')
\bigg \vert
(L(x)=L(x')) \land (L(y)=L(y'))\right]
\end{multline}

By Lemma \ref{length sampler} we know that:
\begin{multline}
\label{eq 1-(27)}
\left(1-\frac{\epsilon}{70}\right)
Pr_{\mathcal{B}(e)}[(L(x)=L(x')) \land (L(y)=L(y'))]
\\
\leq
Pr_{D}\left[(L(x)=L(x')) \land (L(y)=L(y'))
\bigg \vert
(e \in P_1) \land \overline{E_1} \land \overline{E_2}\right]
\\
\leq
\left(1+\frac{\epsilon}{70}\right)
Pr_{\mathcal{B}(e)}[(L(x)=L(x')) \land (L(y)=L(y'))]
\end{multline}
Considering Equations (\ref{eq 1-(25)}), (\ref{eq 1-(26)}) and (\ref{eq 1-(27)}) together, for the lemma to be true, it is enough that:

\begin{multline}
\label{eq 1-(28)}
Pr_{D}[(P_1=P_1') \land (P_2=P_2')|e \in P_1 \land \overline{E_1} \land \overline{E_2} \land (L(x)=L(x')) \land (L(y)=L(y'))]
\\=
Pr_{\mathcal{B}(e)}[(P_1=P_1') \land (P_2=P_2')|(L(x)=L(x')) \land (L(y)=L(y'))]
\end{multline}
Since $e$ is in $M \cap E_I$, if $e \in P_1$ then $E_2$ cannot happen. If $L(x') \geq L(v_1)+t_2$ or $L(x') <L(v_1)$ both sides are zero. Otherwise, $E_1$ cannot happen either and Equation (\ref{eq 1-(28)}) is equivalent to:

\begin{multline}
\label{eq 1-(29)}
Pr_{D}[(P_1=P_1') \land (P_2=P_2')|(e \in P_1) \land (L(x)=L(x')) \land (L(y)=L(y'))]
\\=
Pr_{\mathcal{B}(e)}[(P_1=P_1') \land (P_2=P_2')|(L(x)=L(x')) \land (L(y)=L(y'))]
\end{multline} 
We first argue that:
\begin{multline}
\label{eq 1-(30)}
Pr_{D}[P_1=P_1'|(e \in P_1) \land (L(x)=L(x')) \land (L(y)=L(y'))]
=
Pr_{\mathcal{B}(e)}[P_1=P_1'|(L(x)=L(x')) \land (L(y)=L(y'))]
\end{multline}
This comes from the fact that in both $D$ and $\mathcal{B}(e)$ the distribution of $P_1$ is symmetric with respect to exchange of coordinates. Now we argue that:

\begin{multline}
\label{eq 1-(31)}
Pr_{D}[P_2=P_2'|(e \in P_1) \land (L(x)=L(x')) \land (L(y)=L(y')) \land (P_1=P_1')]
\\=
Pr_{\mathcal{B}(e)}[P_2=P_2'|(L(x)=L(x')) \land (L(y)=L(y')) \land (P_1=P_1')]
\end{multline}
This is true because given $P_1$, both $D$ and $\mathcal{B}(e)$ choose $P_2$ the same way. Combining Equation (\ref{eq 1-(31)}) with Equation (\ref{eq 1-(30)}) we get Equation (\ref{eq 1-(29)}) which completes the proof of correctness. 

The claims about run-time follow from the run-time guarantee on $\mathcal{W}$.
\end{proof}
\subsection{The noise sensitivity estimation algorithm.}
\label{the ns estimation algo}
Our algorithm is:
\hfill \newline
\noindent\fbox{
    \parbox{\textwidth}{
\textbf{Algorithm for estimating noise sensitivity.} (given oracle access to a monotone function $f$\blue{$:\{0,1\}^n \rightarrow \{0,1\}$})
\setlist{nolistsep}
\begin{enumerate}[noitemsep] 
\item \blue{Using the algorithm from \cite{num21} as described in Theorem \ref{estimate influence}, compute an approximation to the influence of $f$ to within a multiplicative factor of ($1 \pm \epsilon/33$)}. This gives us $\tilde{I}$. 
\item Compute $\tilde{p}_A:=\delta \tilde{I}/2$.
\item Initialize $\alpha:=0$ and $\beta:=0$. Repeat the following until $\alpha=\frac{768 \ln 200}{\epsilon^2}$.
\begin{itemize}
\item Use algorithm $\mathcal{A}$ from Lemma \ref{lemma edge sampler} repeatedly to successfully sample an edge $e$.
\item From Lemma \ref{lemma 3} use the algorithm $\mathcal{B}$, giving it $e$ as input, and sample $P_1$, $P_2$, $x$, $y$ and $z$.
\item If it is the case that $f(x)=1$ and $f(z)=0$, then $\alpha:=\alpha+1$.
\item $\beta:=\beta+1$. 
\end{itemize}
\item Set $\tilde{p}_B=\frac{\alpha}{\beta}$.
\item Return $2 \tilde{p}_A \tilde{p}_B$.
\end{enumerate}
}}

We analyze the algorithm by combining all the lemmas from the previous section. First, we prove that:

\begin{lemma}
\label{bound for p_A 1}
It is the case that:
$$
\left(1-\frac{\epsilon}{150}\right) \delta I[f]/2 
\leq
p_A
\leq
\left(1+\frac{\epsilon}{150}\right)\delta I[f] /2
$$

\end{lemma}
\begin{proof}
Summing Lemma \ref{lemma edges uniform} over all the edges in $E_I \cap M$ we get that:

$$
\left(1-\frac{\epsilon}{310}\right)\frac{\delta \cdot |E_I \cap M|}{2^{n}}
\leq
\sum_{e \in E_I \cap M}
p_{e}
\leq
\left(1+\frac{\epsilon}{310}\right)\frac{\delta \cdot |E_I \cap M|}{2^{n}}
$$

Substituting Equation (\ref{eq 1-(6)}) we get:

$$
\left(1-\frac{\epsilon}{310}\right) \frac{\delta I|E_I \cap M|}{2|E_I|}
\leq
p_A
\leq
\left(1+\frac{\epsilon}{310}\right)\frac{ \delta I|E_I \cap M|}{2|E_I|}
$$

With Observation \ref{M is large} this implies the lemma.

\end{proof}

Now, we proceed to prove that $\tilde{p}_A$ and $\tilde{p}_B$ are good approximations for $p_A$ and $p_B$ respectively.
\begin{corollary}
\label{bound for p_A 2}
With probability at least $0.99$:

$$
\left(1-\frac{\epsilon}{16}\right)
\tilde{p}_A
\leq
p_A
\leq
\left(1+\frac{\epsilon}{16}\right)\tilde{p}_A
$$
\end{corollary}
\begin{proof}
From the correctness of the influence estimation algorithm we have that with probability at least $0.99$:

$$
\left(1-\frac{\epsilon}{33}\right) \tilde{I}
\leq
I[f]
\leq
\left(1+\frac{\epsilon}{33}\right) \tilde{I}
$$

Which together with our lemma implies that:

$$
\left(1-\frac{\epsilon}{16}\right)\frac{\delta \tilde{I} }{2}
\leq
p_A
\leq
\left(1+\frac{\epsilon}{16}\right)\frac{ \delta \tilde{I} }{2}
$$
\end{proof}
\begin{definition}
We call an iteration of the main loop \textbf{ successful} if $f(x)=1$ and $f(z)=0$. We also denote by $\phi$ the probability of any given iteration to be successful.  
\end{definition}
In the following two lemmas, we show that $\tilde{p}_B$ is a good approximation to $\phi$ and that $\phi$, in turn, is a good approximation to $p_B$. This will allow us to conclude that $\tilde{p}_B$ is a good approximation to $p_B$.
\begin{lemma}
\label{bound for phi}
With probability at least $0.99$ we have that:
$$
\left(1-\frac{\epsilon}{16}\right)\tilde{p}_B
\leq
\phi
\leq
\left(1+\frac{\epsilon}{16}\right)\tilde{p}_B
$$
Additionally, the expected number of iterations of the main loop is  $O(1/(\phi \epsilon^2))$.
\end{lemma}
\begin{proof}
See Appendix C, Subsection \ref{proof 3}.
\end{proof}

\begin{lemma}
\label{bound for p_B}

It is the case that:

$$
\left(1-\frac{\epsilon}{16}\right) \phi
\leq
p_B
\leq
\left(1+\frac{\epsilon}{16}\right)\phi
$$
\end{lemma}
\begin{proof}

See Appendix C, Subsection \ref{proof 4}.
\end{proof}

Corollary \ref{bound for p_A 2} with Lemmas  \ref{bound for phi} and \ref{bound for p_B} together imply that with probability at least $2/3$:

$$
\left(1-\frac{\epsilon}{5} \right)\tilde{p}_A \tilde{p}_B
\leq 
p_A p_B
\leq
\left(1+\frac{\epsilon}{5} \right)\tilde{p}_A \tilde{p}_B
$$
Combining this with Lemma \ref{corner cutting} and  Equation (\ref{eq 1-(5)}) we get that:

$$
\left(1-\frac{\epsilon}{2} \right)2\tilde{p}_A \tilde{p}_B
\leq 
NS_\delta[f]
\leq
\left(1+\frac{\epsilon}{2}\right)2\tilde{p}_A \tilde{p}_B
$$
This proves the correctness of the algorithm. Now consider the number of queries:
\begin{itemize}
\item Estimating the influence requires $O\left(\frac{\sqrt{n} \log(n/\epsilon)}{I[f] \epsilon^3}\right)$ queries and polynomial time. Since by Observation \ref{epsilon is large} it is the case that $\epsilon \geq 1/n$, this is at most $O\left(\frac{\sqrt{n} \log(n)}{I[f] \epsilon^3}\right)$.
\item By Lemma \ref{lemma edge sampler}, successfully sampling an edge requires $O \left(\frac{\sqrt{n} \log^{1.5} n}{I[f] \epsilon}\right)$ queries and polynomial time. By Lemma \ref{lemma 3} for each edge we willl additionally spend a polynomial amount of extra time.
\item By Lemmas \ref{bound for phi} and \ref{bound for p_B} we will have $O \left(\frac{1}{\phi \epsilon^2} \right)=O \left(\frac{1}{p_B \epsilon^2} \right)$ iterations of the loop. 
\end{itemize}

Therefore, the overall run-time is polynomial, and the overall number of queries made is:

$$
O\left(\frac{\sqrt{n} \log n}{I[f] \epsilon^3}\right)
+
O \left(\frac{\sqrt{n} \log^{1.5} n}{I[f] \epsilon}\right)
O \left(\frac{1}{p_B \epsilon^2} \right)=O \left(\frac{\sqrt{n} \log^{1.5} n}{I[f] \epsilon} \cdot \frac{1}{p_B \epsilon^2} \right)
$$
Now using Lemmas \ref{bound for p_A 1} and \ref{corner cutting} together with Equation (\ref{eq 1-(5)}) it equals to:

$$
O\left(\frac{\sqrt{n} \delta \log^{1.5} n}{p_A \epsilon} \cdot \frac{1}{p_B \epsilon^2} \right)
=
O\left(
 \frac{\sqrt{n} \delta \log^{1.5} n}{NS_\delta[f] \epsilon^3}
 \right) 
$$
Which is the desired number of queries.

\section{Lower bounding the sample complexity.}

\blue{We begin our proof of Theorems \ref{2-1.3}, \ref{2-1.1} and \ref{2-1.2} by first defining the distributions} $D_1^{B}$, $D_1^{I}$ and $D_1^{\delta}$ \blue{to} consist of a single function $f_0$\blue{$:\{0,1\}^n\rightarrow\{0,1\}$}:
$$
f_0(x)
=\begin{cases} 1 &\mbox{if } L(x) > n/2+k \sqrt{n \log n} \\
0 & \mbox{ otherwise} \end{cases}$$
And here $k$ is chosen so that $n/2+k \sqrt{n \log n}$ is the smallest integer larger than $n/2$ for which $B[f_0] \leq 1/n^{C_1}$.

For our lower bounds we will need to show that $f_0$ and ORs of $f_0$ with other functions have useful properties. For this we will need the following lemma. Informally, it says that $B[F]$, $I[F]$ and $NS_{\delta}[F]$ are continuous functions of $F$ in the sense that changing $F$ a little bit does not change them drastically.

\begin{lemma}
\label{2-1.5}
For any monotone function $F$\blue{$:\{0,1\}^n\rightarrow\{0,1\}$}, denote by $F'$ the OR of $F$ and $f_0$. Then:
\begin{itemize}
\item[a)] 
$$
B[F] \leq B[F'] \leq B[F']+\frac{1}{n^{C_1}}
$$
\item[b)]
$$
\bigg \lvert 
I[F]- I[F'] 
\bigg \rvert\leq \frac{2n}{n^{C_1}}
$$
\item[c)] For any $\delta$:
$$
\bigg \lvert 
NS_{\delta}[F]- NS_{\delta}[F'] 
\bigg \rvert\leq \frac{2}{n^{C_1}}
$$
\end{itemize}
\end{lemma}
\begin{proof}
See Appendix D, Subsection \ref{proof 5}.
\end{proof}
\begin{lemma}
\label{2-1.4}
It is the case that:
\begin{itemize} 
\item[a)] $k 
\leq
\sqrt{
\frac{C_1}{8}}$, and hence $k$ is also a constant.
\item[b)] $B[f_0] = 1/\Theta(n^{C_1})$
\item[c)]$\Omega( 1/n^{C_1}) \leq I[f_0] \leq O(n/n^{C_1})$
\item[e)] For any $\delta$, it is the case that $\Omega( 1/n^{C_1+1}) \leq NS_{\delta}[f_0] \leq O(1/n^{C_1})$.
\end{itemize}
\end{lemma}
\begin{proof}
See Appendix C, Subsection \ref{proof 6}.
\end{proof}\newline
We will use the two following main lemmas, that we will prove in two separate subsections. The first one is a computational lemma that says that any function family that is ``thin" at the level $L(x)=n/2+k \sqrt{n \log n}$ can be transformed into two distributions that are hard to distinguish. By ``thin" we mean that they have few positive points at the level right below the threshold of $f_0$. The second lemma says that there exist function families that are both ``thin" in the sense the first lemma requires and have a large amount of bias, influence and noise sensitivity.
\begin{lemma}
\label{2-1.6}
Suppose $F$\blue{$:\{0,1\}^n\rightarrow\{0,1\}$} is a monotone Boolean function with the property that:
$$Pr_{x \in_R \{0,1\}^n}[F(x)=1 \vert L(x)=n/2+k \sqrt{n \log n}] \leq 1/q_0$$
Additionally, suppose $\mathcal{C}$ is an algorithm that makes $o(q_0)$ queries given access to the following random function:
\begin{itemize}
\item With probability $1/2$ it is drawn from $D_1$ that consists of only $f_0$.
\item With probability $1/2$ it is drawn from $D_2$ that consists of an OR of the following:
\begin{itemize}
\item The function $f_0$.
\item $F(\sigma(x))$, where $\sigma$ is a random permutation of indices.
\end{itemize}
\end{itemize}
Consequently, suppose that $\mathcal{C}$ outputs a guess whether its input was from $D_1$ or $D_2$. Then $\mathcal{C}$ has to err with probability more than $1/3$. 
\end{lemma}

\begin{lemma}
\label{2-1.7}
There exist functions $F^{B}$\blue{$:\{0,1\}^n\rightarrow\{0,1\}$}, $F^{I}$\blue{$:\{0,1\}^n\rightarrow\{0,1\}$} and for every $1/n \leq\delta \leq 1/2$ there exists $F^{\delta}$\blue{$:\{0,1\}^n\rightarrow\{0,1\}$} such that:
\begin{itemize}
\item Any $f$ in $\{F^B, F^I\} \cup \{F^{\delta}:1/n \leq\delta \leq 1/2\}$ has the property that:
$$Pr_{x \in_R \{0,1\}^n}[f(x)=1 \vert L(x)=n/2+k \sqrt{n \log n}] \leq \Theta\left(1/n^{C_2} \cdot e^{\sqrt{C_1 \log n/2}} \right)$$
\item $B[F^{B}] \geq \Omega(1/n^{C_2})$.
\item $I[F^{I}] \geq \Omega(\sqrt{n}/n^{C_2})$.
\item For any $1/n \leq \delta \leq 1/2$,  it is the case that:
\begin{equation*}
NS_{\delta}[F^{\delta}]
\geq
\begin{cases}
\Omega(\delta \sqrt{n} /n^{C_2})
&\text{ if }1/n \leq \delta \leq 1/\sqrt{n} \\
\Omega(1/n^{C_2})
&\text{ if }1/\sqrt{n} < \delta \leq 1/2
\end{cases}
\end{equation*}
\end{itemize}
\end{lemma}

Recall that we defined $D_1^B$, $D_1^I$ and $D_1^\delta$ to be composed of only the function $f_0$. Now we also have functions that are thin and have large bias, influence and noise sensitivity. Therefore, we use them to define distributions we can use with Lemma \ref{2-1.6}:
\begin{itemize}
\item $D_2^B$ as the OR of $f_0(x)$ and $F^B(\sigma(x))$. Where, recall, $\sigma$ is a random permutation of indices.
\item $D_2^I$ as the OR of $f_0(x)$ and $F^I(\sigma(x))$. 
\item For each $1/n \leq \delta \leq 1/2$, we define $D_2^\delta$ as the OR of $f_0(x)$ and $F^\delta(\sigma(x))$. 
\end{itemize} 
\begin{observation}
\label{2-1.9}
 Permuting the indices to an input of a function preserves its bias, influence and noise sensitivity. That, together with
Lemma \ref{2-1.5} and Lemma \ref{2-1.7}, implies that:
\begin{itemize}
\item[a)] For any $f$ in $D_2^{B}$, we have $B[f]\geq \Omega(1/n^{C_2})$.
\item[b)] For any $f$ in $D_2^{I}$, we have $I[f]\geq \Omega(\sqrt{n}/n^{C_2}) - O(n/n^{C_1}) = \Omega(\sqrt{n}/n^{C_2})$. The last equality is true because $C_1-1 > C_2$.
\item[c)]For all $1/n \leq \delta \leq 1/2$ and for all $f$ in 
$D_2^\delta$: 

$$NS_{\delta}[f] \geq \begin{cases}
\Omega(\delta \sqrt{n} /n^{C_2}) - O(1/n^{C_1})
&\text{ if }1/n \leq \delta \leq 1/\sqrt{n} \\
\Omega(1/n^{C_2}) - O(1/n^{C_1})
&\text{ if }1/\sqrt{n} < \delta \leq 1/2
\end{cases}
=
\begin{cases}
\Omega(\delta \sqrt{n} /n^{C_2}) 
&\text{ if }1/n \leq \delta \leq 1/\sqrt{n} \\
\Omega(1/n^{C_2})
&\text{ if }1/\sqrt{n} < \delta \leq 1/2
\end{cases}
$$
Here, again, the last equality is true because $C_1 -1 > C_2$
\end{itemize}
\end{observation}

\blue{
Now, we are ready to prove our main theorems.  Let us first consider the case of estimating the bias. We will prove it by contradiction, showing that the negation of Theorem \ref{2-1.1} implies we can reliably distinguish two distributions that Lemma \ref{2-1.6} prevents us from distinguishing. Suppose $\mathcal{L}_B$ is an algorithm, taking as input a monotone function $f: \{0,1\}^n\rightarrow \{0,1\}$ and with probability at least $2/3$ distinguishing whether $f$ (i) has a bias of $\Theta(1/n^{C_1})$ or (ii) has a bias of at least $\Omega(1/n^{C_2})$. For the sake of contradiction, assume that it makes $o \left(\frac{n^{C_2}}{e^{\sqrt{C_1 \log n/2}}} \right)$ queries.}

\blue{
By Lemma \ref{2-1.4} every function in $D_1$ has a bias of $\Theta(1/n^{C_1})$ and by Observation \ref{2-1.9}, the bias of every function in $D_2^B$ is at least $\Omega(1/n^{C_2})$. Therefore, $\mathcal{L}_B$ can distinguish between them with probability at least $2/3$ making $o \left(\frac{n^{C_2}}{e^{\sqrt{C_1 \log n/2}}} \right)$. But by Lemma \ref{2-1.6}, such an algorithm has to err with probability more than $1/3$. We have a contradiction and Theorem \ref{2-1.1} follows.
}

\blue{
Theorems \ref{2-1.2} and \ref{2-1.3} follow analogously.
}

\subsection{Proof of Lemma \ref{2-1.6}}

Suppose $\mathcal{C}$ is an adaptive algorithm that makes at most $q$ queries and distinguishes a random function in $D_2$ from $D_1$. We denote the number of queries it makes as $q$. Without loss of generality, we assume that it always makes precisely $q$ queries. The algorithm is adaptive and in the end it outputs $1$ or $2$. Then, the probability that the algorithm correctly distinguishes equals:
$$
p_{\mathcal{C}} \myeq
\frac{1}{2}Pr[\mathcal{C} \text{ returns } 1
\vert f \in_{R} D_1]
+
\frac{1}{2}
Pr[\mathcal{C} \text{ returns } 2
\vert f \in_{R} D_2]
$$
We call the difference above \textbf{the distinguishing power} of $\mathcal{C}$.
\begin{observation}
\label{2-2.1}
For any $f$ in $D_1 \cup D_2$ it is the case that if $L(x)>n/2+k\sqrt{n \log n}$, then $f(x)=1$. Therefore, without loss of generality we can assume that $\mathcal{C}$ never queries any point in that region.  
\end{observation}
\begin{observation}
\label{2-2.2}
If $\mathcal{C}$ is randomized, we can think of it as a probability distribution over deterministic algorithms. The distinguishing powers of $\mathcal{C}$ then will be the weighted sum of the distinguishing power of the deterministic algorithms, weighted by their probabilities. Therefore, the distinguishing power of the best of these deterministic algorithms is at least that of $\mathcal{C}$. Thus, without loss of generality we can assume that $\mathcal{C}$ is deterministic.
\end{observation}
Now, since $\mathcal{C}$ is deterministic and it makes $q$ queries, it can be represented as a decision tree of depth $q$. At each node, $\mathcal{C}$ queries a point and proceeds to the next node. In the end, after $q$ queries, the algorithm reaches a leaf and outputs the label of the leaf, namely $1$ or $2$. We can divide this decision tree into two regions:
\begin{itemize}
\item A path $A$ that the algorithm takes if at every query it receives zero.
\item The rest of the decision tree. We call this region $B$.
\end{itemize}

By Observation \ref{2-2.1} and the definition of $f_0$, when the algorithm is given access to a member of $D_1$, namely $f_0$, all the points $x$ that it queries will have $f(x)=f_0(x)=0$. Therefore, the algorithm will follow the path $A$ on the decision tree and end up on the single leaf there.

Suppose now the algorithm is given access to a function in $D_2$. There are two cases:
\begin{itemize}
\item All the queries $x^j$ it makes, will have $f(x^j)=0$. Then, on the decision tree it will follow the path $A$ and end up at the single leaf on it.
\item After making query $x^{j}$ for which $f(x^j)=1$, the algorithm ends up in the subset of the tree we call $B$. We call the probability that this happens $p_{\text{get1}}$.
\end{itemize}

If the single leaf in $A$ is not labeled with $1$, then the algorithm will always err, given access to $D_1$. Similarly, the algorithm can reach a leaf in $B$ only if it was given access to $D_2$. Thus, labeling all such leaves with $2$ can only increase the distinguishing power. Therefore, without loss of generality, we assume that this is the labeling used in $\mathcal{C}$. Then, the distinguishing power $p_{\mathcal{C}}$ equals $\frac{1}{2}(1+p_{\text{get}1})$.

The following lemma shows that we can assume that the algorithm is non-adaptive and only makes queries on the level $n/2+k\sqrt{n \log n}$.
\begin{lemma}
\label{2-2.3}
There exists an algorithm $\mathcal{D}$ that satisfies all of the following:
\begin{itemize}
\item It is deterministic and \textbf{non}-adaptive.
\item Its distinguishing power between $D_1$ and $D_2$ is at least $p_{\mathcal{C}}$.
\item Just as $\mathcal{C}$ it makes $q$ queries. We call them $z^1,...,z^q$.
\item For each of these $z^j$, it is the case that $L(z^j)=n/2+k\sqrt{n \log n}$.
\end{itemize}
\end{lemma}
\begin{proof}
Consider the queries that $\mathcal{C}$ makes along the path $A$. Call them $y^1,...,y^q$. We define $z^j$ as an arbitrary point that satisfies (i) $L(z^j)=n/2+k\sqrt{n \log n}$ and (ii) $ y^j\preceq z^j$. At least one such point has to exist since $L(y^j) \leq n/2+k\sqrt{n \log n}$.

The algorithm $\mathcal{D}$ queries each of these $z^j$ and returns $2$ if for at least one of them $f(z^j)=1$. Otherwise it returns $1$.

Since $ y^j\preceq z^j$ and the functions are monotone, whenever $f(y^j)=1$, then $f(z^j)=1$. At least one of $f(y^j)$ equals one with probability $p_{\text{get1}}$, and therefore at least one of $f(z^j)$ equals one with probability at least $p_{\text{get1}}$.

Thus, $\mathcal{D}$ has a distinguishing power of at least $\frac{1}{2}(1+p_{\text{get1}})$. This implies that the distinguishing power of $\mathcal{D}$ is at least that of $\mathcal{C}$.
\end{proof}

Now, we can bound the distinguishing power of $\mathcal{D}$, which we call $p_{\mathcal{D}}$. We have:

\begin{multline*}
\frac{1}{2}
Pr[\mathcal{D} \text{ returns } 2
\vert f \in_{R} D_2]
=
Pr_{f \in_{R} D_2}\left[\bigvee_{j=0}^q f(z^j)=1\right]
\leq
\sum_{j=0}^{q}
Pr_{f \in_{R} D_2}\left[f(z^j)=1\right]
\\=
\sum_{j=0}^{q}
Pr_{\sigma \text{ is a random permutation}}\left[F(\sigma( z^j))=1\right]
=
q \cdot Pr_{x \in_{R} \{0,1\}^n
}\left[F(x)=1
\bigg \lvert
L(x)=n/2+k \sqrt{n \log n}\right
]
\leq
\frac{q}{q_0}
\end{multline*}
Above we used (i) a union bound and the fact that by Lemma \ref{2-2.3}, the algorithm $\mathcal{D}$ is non-adaptive (ii) The fact that by definition $f(x)=F(\sigma (x))$. Recall that $\sigma$ is the random permutation $F$ was permuted with. (iii) For any constant $x$, $\sigma(x)$ is uniformuly distributed among the vertices with the same level. (iv) Lemma \ref{2-2.3} together with the condition on the function $F$.

Therefore, we get that:

$$
p_{\mathcal{D}}
\myeq
\frac{1}{2}
Pr[\mathcal{D} \text{ returns } 2
\vert f \in_{R} D_2]
+\frac{1}{2}
Pr[\mathcal{D} \text{ returns } 1
\vert f \in_{R} D_1]
\leq
\frac{q}{q_0}+\frac{1}{2}
$$

Since $q=o(q_0)$, then $p_{\mathcal{D}}$ has to be less than $2/3$. Since $p_{\mathcal{C}}$ is at most $p_{\mathcal{D}}$ by Lemma \ref{2-2.3}, then $p_{\mathcal{C}}$ also has to be less than $2/3$. This proves the lemma.
\subsection{Proof of Lemma \ref{2-1.7}}
Consider the distribution\footnote{There are two differences between this distribution and Talagrand random functions: (i) Here we choose $1/n^{C_2}\cdot 2^{\sqrt{n}}$ clauses, whereas Talagrand functions have just $2^{\sqrt{n}}$ clauses. (ii) In Talagrand functions the indices in each clause are sampled with replacement, whereas here we sample them without replacement.} $H$ of functions, which is OR of $1/n^{C_2} \cdot 2^{\sqrt{n}}$ AND clauses of uniformly and independently chosen subsets of $\sqrt{n}$ indices, chosen without replacement.
We will prove that:
\begin{itemize}
\item[a)] \textbf{Any} $f$ in $H$ has the property that:
$$Pr_{x \in_R \{0,1\}^n}[f(x)=1 \vert L(x)=n/2+k \sqrt{n \log n}] \leq \Theta\left(1/n^{C_2} \cdot e^{\sqrt{C_1 \log n/2}} \right)$$
\item[b)] $E_{f \in_{R} H} [B[f]] \geq \Omega(1/n^{C_2})$.
\item[c)] $E_{f \in_{R} H} [I[f]] \geq \Omega(\sqrt{n}/n^{C_2})$.
\item[d)] For any $1/n \leq \delta \leq 1/2$,  it is the case that:
\begin{equation*}
E_{f \in_{R} H}[NS_{\delta}[f]]
\geq
\begin{cases}
\Omega(\delta \sqrt{n} /n^{C_2})
&\text{ if }1/n \leq \delta \leq 1/\sqrt{n} \\
\Omega(1/n^{C_2})
&\text{ if }1/\sqrt{n} < \delta \leq 1/2
\end{cases}
\end{equation*}
\end{itemize}
Then, the corresponding claims of the lemma will follow by an application of the probabilistic method. We have divided into subsections the proofs of the claims above.
\subsection{Proof of (a)}

Here we treat the clauses as fixed and look at the probability over the randomness of choosing $x$ to satisfy any of them. For a given clause, the probability that $x$ will satisfy it equals:

\begin{multline*}
\prod_{i=0}^{\sqrt{n}-1}
\frac{n/2+k \sqrt{n \log n}-i}{n}
\leq
\left(
\frac{n/2+k \sqrt{n \log n}}{n}
\right)^{\sqrt{n}}
\\=
\frac{1}{2^{\sqrt{n}}}
\left(
1+\frac{2k \sqrt{ \log n}}{\sqrt{n}}
\right)^{\sqrt{n}}
\leq
\frac{1}{2^{\sqrt{n}}}
\Theta\left(e^{2k \sqrt{\log n}}\right)
\leq
\frac{1}{2^{\sqrt{n}}}
\cdot 
\Theta
(
e^{\sqrt{C_1 \log n/2}}
)
\end{multline*}
In the very end we used that by Lemma \ref{2-1.4} it is the case that $k \leq \sqrt{C_1 /8}$.

Now, that we know the probability for one clause, we can upper-bound the probability $x$ satisfies any of the $\frac{1}{n^{C_2}}\cdot 2^{\sqrt{n}}$ using a union bound. This gives us an upper bound of $\Theta\left(1/n^{C_2} \cdot e^{\sqrt{C_1 \log n/2}} \right)$.
\subsection{Proof of b)}

It is the case that:

$$
E_{f \in_{R} H}[B[f]]
=
E_{f \in_{R} H}[E_{x \in_{R} \{0,1\}^n}[f(x)]]
=
E_{x \in_{R} \{0,1\}^n}[E_{f \in_{R} H}[f(x)]]
\geq
\frac{1}{2}
E_{x \in_{R} \{0,1\}^n}\left[E_{f \in_{R} H}[f(x)] \Bigg\vert L(x) \geq \frac{n}{2}\right]
$$

If we fix a value of $x$ for which $L(x)\geq n/2$, and randomly choose a single AND of $\sqrt{n}$ indices, the probability that it evaluates to one on $x$ is:

\begin{multline*}
\frac{L(x)}{n}
\cdot
\frac{L(x)-1}{n-1}
\cdot
...
\cdot
\frac{L(x)-\sqrt{n}+1}{n-\sqrt{n}+1}
\geq
\left(
\frac{L(x)-\sqrt{n}+1}{n-\sqrt{n}+1}
\right)^{\sqrt{n}}
\geq 
\left(
\frac{n/2-\sqrt{n}+1}{n-\sqrt{n}+1}
\right)^{\sqrt{n}}
\\ \geq
\left(
\frac{1}{2}
-\frac{\sqrt{n}/2-1/2}{n-\sqrt{n}+1}
\right)^{\sqrt{n}}
\geq
\frac{1}{2^{\sqrt{n}}}
\left(
1
-\frac{1}{\Theta(\sqrt{n})}
\right)^{\sqrt{n}}
=\frac{1}{\Theta(2^{\sqrt{n}})}
\end{multline*}
Then, since we have $2^{\sqrt{n}}/n^{C_2}$ clauses and they are chosen independently:

$$
E_{x \in_{R} \{0,1\}^n}\left[E_{f \in_{R} H}[f(x)] \Bigg\vert L(x) \geq \frac{n}{2}\right]
\geq
1-\left(
1-\frac{1}{\Theta(2^{\sqrt{n}})}
\right)^{2^{\sqrt{n}}/n^{C_2}}
\geq
\frac{1}{\Theta(n^{C_2})}
$$
This implies that $E_{f \in_{R} H}[B[f]] \geq 1/\Theta(n^{C_2})$.

\subsection{Proof of c)}
We have that:
\begin{equation}
\label{eq 2-(3)}
E_{f \in_{R} H}[I[f]]
=
E_{f \in_{R} H}
\left[
n \cdot 
Pr_{x \in_{R} \{0,1\}^n, i \in_{R} [n]}
[
f(x) \neq f(x^{\oplus i})
]
\right]
=
n \cdot
Pr_{f \in_{R} H, x \in_{R} \{0,1\}^n, i \in_{R} [n]}
[
f(x) \neq f(x^{\oplus i})
]
\end{equation}
The probability of an event is the expectation of its indicator random variable. Using this twice, gives us the second equality above.

From Hoeffding's inequality, it follows that with probability at least $0.95$ it is the case that $n/2-\sqrt{n} \leq L(x) \leq n/2+\sqrt{n}$. From this and Equation (\ref{eq 2-(3)}) it follows:
\begin{multline}
\label{eq 2-(4)}
E_{f \in_{R} H}[I[f]]
\geq
n \cdot
Pr_{f \in_{R} H, x \in_{R} \{0,1\}^n, i \in_{R} [n]}
\left[
f(x) \neq f(x^{\oplus i}) \bigg \vert
n/2-\sqrt{n} \leq L(x) \leq n/2+\sqrt{n}
\right]
\\
\times
Pr_{x \in_{R} \{0,1\}^n}
[
n/2-\sqrt{n} \leq L(x) \leq n/2+\sqrt{n}
]
\\
\geq
0.95n
Pr_{f \in_{R} H, x \in_{R} \{0,1\}^n, i \in_{R} [n]}
\left[
f(x) \neq f(x^{\oplus i}) \bigg \vert
n/2-\sqrt{n} \leq L(x) \leq n/2+\sqrt{n}
\right]
\end{multline}

Now we will lower-bound $Pr_{f \in_{R} H, i \in_{R} [n]}
[
f(x) \neq f(x^{\oplus i})
]
$ for any $x$, for which $n/2-\sqrt{n} \leq L(x) \leq n/2+\sqrt{n}$. Name the clauses in $f$ as $\land_1, \land_2,...,\land_{\frac{1}{n^{C_2}}2^{\sqrt{n}}}$.
For any clause $\land_{j}$ we have:
$$
Pr_{f \in_{R} H}[\land_{j} \text{ is satisfied}]
=
\frac{L(x)}{n}
\cdot
\frac{L(x)-1}{n-1}
\cdot
...
\cdot
\frac{L(x)-\sqrt{n}+1}{n-\sqrt{n}+1}
$$
And since $n/2-\sqrt{n} \leq L(x) \leq n/2+\sqrt{n}$:
$$
\frac{1}{2^{\sqrt{n}}}
\left(
1-\frac{1}{\Theta(\sqrt{n})}
\right)^{\sqrt{n}}
\leq
\left(
\frac{L(x)-\sqrt{n}+1}{n-\sqrt{n}+1}
\right)^{\sqrt{n}}
\leq
Pr_{f \in_{R} H}[\land_{j} \text{ is satisfied}]
\leq
\left(
\frac{L(x)}{n}
\right)^{\sqrt{n}}
\leq
\frac{1}{2^{\sqrt{n}}}
\left(
1+\frac{1}{\Theta(\sqrt{n})}
\right)^{\sqrt{n}}
$$
This implies that:
\begin{equation}
\label{eq 2-(5)}
Pr_{f \in_{R} H}[\land_{j} \text{ is satisfied}]
=\frac{1}{\Theta(2^{\sqrt{n}})}
\end{equation}
For every $i$, so that $1\leq i \leq 2^{\sqrt{n}}/n^{C_2}$, consider the following sequence of events, which we call $M_i$:
\begin{itemize}
\item $x$ satisfies $\land_i$. By Equation (\ref{eq 2-(5)}) the probability of this happening is $1/\Theta(2^{\sqrt{n}})$.
\item $x$ does not satisfy all the other $\frac{1}{n^{C_2}} 2^{\sqrt{n}}-1$ clauses. Since they are are chosen independently, by Equation (\ref{eq 2-(5)}) we have that the probability of this is $\left(1-1/\Theta(2^{\sqrt{n}}) \right)^{2^{\sqrt{n}}/n^{C_2}-1}$, which is at least $\Theta(1)$.
\item $i$ is one of the inputs that are relevant to $\land_i$. The probability of this is $1/\sqrt{n}$.
\end{itemize}
Since these three events are independent:
\begin{equation*}
Pr_{f \in_{R} H}[M_i]
\geq\frac{1}{\Theta(2^{\sqrt{n}})}
\cdot \Theta(1) \cdot \frac{1}{\sqrt{n}}
=\frac{1}{\Theta(\sqrt{n}2^{\sqrt{n}})}
\end{equation*}

If $M_i$ happens, then $f(x) \neq f(x^{\oplus i})$. Additionally for different values of $i$, the $M_i$ are disjoint and the probability of $M_i$ is the same for all $i$ by symmetry. Thus we have:
\begin{multline}
\label{eq 2-(6)}
Pr_{f \in_{R} H, i \in_{R} [n]}
\left[f(x) \neq f(x^{\oplus i})\right]
\geq
Pr_{f \in_{R} H, i \in_{R} [n]}
\left[
\bigvee_i M_i
\right]
=
\frac{2^{\sqrt{n}}}{n^{C_2}}
Pr_{f \in_{R} H, i \in_{R} [n]}
\left[
M_1
\right]
\\ \geq \frac{2^{\sqrt{n}}}{n^{C_2}}
\cdot \frac{1}{\Theta(\sqrt{n}2^{\sqrt{n}})}
=
\frac{1}{\Theta(\sqrt{n} \cdot n^{C_2})}
\end{multline}

Combining Equations (\ref{eq 2-(4)}) and (\ref{eq 2-(6)}):
$$
E_{f \in_{R} H}[I[f]]
\geq
0.95n \cdot \frac{1}{\Theta(\sqrt{n} \cdot n^{C_2})}
=\Omega\left(\frac{\sqrt{n}}{ n^{C_2}}\right)
$$

\subsection{Proof of d)}
Recall that in the definition of noise sensitivity, $x$ is chosen uniformly and $y$ is chosen by flipping each bit of $x$ with probability $\delta$. We have:

\begin{equation}
\label{eq 2-(7)}
E_{f \in_{R} H}[NS_{\delta}[f]]
=
E_{f \in_{R} H}
\left[
Pr_{(x,y) \in_{R} T_{\delta}}
[
f(x) \neq f(y)
]
\right]
=
Pr_{f \in_{R} H,(x,y) \in_{R} T_{\delta}}
[
f(x) \neq f(y)
]
\end{equation}

Consider the following three ``good" events, which are similar to the ones introduced in \cite{num9} to analyze the noise sensitivity of Talagrand random functions:
\begin{itemize}
\item $G_1$ is when $n/2-\sqrt{n} \leq L(x) \leq n/2+\sqrt{n}$. By Hoeffding's inequality, its probability is at least $0.95$.
\item $G_2$ is when $n/2-\sqrt{n} \leq L(y) \leq n/2+\sqrt{n}$. Since $y$ is also distributed uniformly, its probability is also at least $0.95$.
\item Denote by $S_x$ the set of indices $i$ for which $x_i=1$. By the definition of noise sensitivity, in expectation $\delta |S_x|$ of them become zero in $y$. The event $G_3$ happens when at least $\delta |S_x|/2$ of them are zero in $y$. By the Chernoff bound, the probability of this is at least $1-\exp(-\delta n/8)\geq 1-\exp(-1/8) \geq 0.11$. 
\end{itemize}
By a union bound, with probability at least $0.01$ all the events $G_1$, $G_2$ and $G_3$ happen. Therefore:

\begin{multline}
\label{eq 2-(8)}
E_{f \in_{R} H}[NS_{\delta}[f]]
\geq
Pr_{f \in_{R} H, (x,y) \in_{R} T_{\delta}}
\left[
f(x) \neq f(y) \bigg \vert
G_1 \land G_2 \land G_3
\right]
\cdot
Pr_{(x,y) \in_{R} T_{\delta}}
[
G_1 \land G_2 \land G_3
]
\\\geq
0.01 \cdot
Pr_{f \in_{R} H, (x,y) \in_{R} T_{\delta}}
\left[
f(x) \neq f(y) \bigg \vert
G_1 \land G_2 \land G_3
\right]
\end{multline}
Now, suppose we are given values of $x$ and $y$ that satisfy $G_1 \land G_2 \land G_3$, we will lower bound $Pr_{f \in_{R} H}
\left[
f(x) \neq f(y)
\right]$. Using $G_1$ and $G_2$, just as in the proof of equation (5), we have that for any clause $\land_j$:
\begin{align}
\label{eq 2-(9)}
Pr_{f \in_{R} H}[\land_{j} \text{ is satisfied by } x]
=\frac{1}{\Theta(2^{\sqrt{n}})}
&&
Pr_{f \in_{R} H}[\land_{j} \text{ is satisfied by } y]
=\frac{1}{\Theta(2^{\sqrt{n}})}
\end{align}

Now, analogous to how we lower bounded the expected influence while proving part c), consider the following sequence of events, which we call $N_i$:
\begin{itemize}
\item $x$ satisfies $\land_i$. By Equation (\ref{eq 2-(9)}), the probability of this is at least $1/\Theta(2^{\sqrt{n}})$.
\item All the $\land_j$ for $j \neq i$ are unsatisfied by both $x$ and $y$. By Equation (\ref{eq 2-(9)}) and a union bound, for each individual clause the probability of being unsatisfied by both $x$ and $y$ is at least $1-2 \cdot 1/\Theta(2^{\sqrt{n}})=1-1/\Theta(2^{\sqrt{n}})$. By independence, the probability of the overall event is at least $\left(1-1/\Theta(2^{\sqrt{n}}) \right)^{2^{\sqrt{n}}/n^{C_2}-1}$, which is at least $\Theta(1)$.
\item Given that $x$ satisfies $\land_i$, it happens that at least one of the coordinates relevant to $\land_i$ is zero in $y$. We call the probability of this happening $p_{\text{flip}}$. 
\end{itemize}
The third event is conditioned on the first one, so the probability that both happen equals to the product of their probabilities. In addition, The first and third event depend only on the randomness in choosing $\land_i$, and the second event only on the randomness in choosing all the other clauses. Therefore the second event is independent from the first and third, and thus:
\begin{equation}
\label{eq 2-(10)}
Pr_{f \in_{R} H}
[N_i]
\geq
\frac{1}{\Theta(2^{\sqrt{n}})}
\cdot
\Theta(1) 
\cdot
p_{\text{flip}}
=
\frac{p_{\text{flip}}}{\Theta(2^{\sqrt{n}})}
\end{equation}

We now lower-bound $p_{\text{flip}}$. Because of $G_3$, at least $\delta |S_x|/2$ of the indices in $S_x$ become zero in $y$. Let $S_{\land_i}$ be the set of $\sqrt{n}$ indices relevant to $\land_i$. Since they were chosen uniformly at random then, conditioning on $\land_i$ being satisfied, $S_{\land_i}$ has equal probability of being any subset of $S_{x}$ of size $\sqrt{n}$. Therefore, the probability that at least one of them ends up among the indices in $S_x$ that become zero in $y$:
\begin{equation*}
p_{\text{flip}}
=
1-\prod_{j=0}^{\sqrt{n}-1}
\left(
1-\frac{\delta |S_x|/2}{|S_x|-j}
\right)
\geq
1-(1-\delta/2)^{\sqrt{n}}
\geq
\begin{cases}
\Theta(\delta \sqrt{n})
&\text{ if }1/n \leq \delta \leq 1/\sqrt{n} \\
\Theta(1)
&\text{ if }1/\sqrt{n} < \delta \leq 1/2
\end{cases}
\end{equation*}
Combining this with Equation (\ref{eq 2-(10)}), we get:
\begin{equation}
\label{eq 2-(11)}
Pr_{f \in_{R} H}
[N_i]
\geq
\frac{p_{\text{flip}}}{\Theta(2^{\sqrt{n}})}
\geq
\begin{cases}
\Theta(\delta \sqrt{n} /2^{\sqrt{n}})
&\text{ if }1/n \leq \delta \leq 1/\sqrt{n} \\
\Theta(1/2^{\sqrt{n}})
&\text{ if }1/\sqrt{n} < \delta \leq 1/2
\end{cases}
\end{equation}
We now combine (i) Equation (\ref{eq 2-(8)}) (ii) The fact that if $N_i$ happens then $f(x) \neq f(y)$ (iii) The fact that the different $N_i$ are disjoint and the probability of $N_i$ is the same for all $i$ by symmetry (iv) Equation (\ref{eq 2-(11)}):
\begin{equation*}
E_{f \in_{R} H}[NS_{\delta}[f]]
\geq
0.01 \cdot
Pr_{f \in_{R} H}
\left[
\bigvee_i N_i
\right]
=
0.01 \cdot
\frac{2^{\sqrt{n}}}{n^{C_2}}
Pr_{f \in_{R} H, i \in_{R} [n]}
\left[
N_1
\right]
\geq
\begin{cases}
\Theta(\delta \sqrt{n} /n^{C_2})
&\text{ if }1/n \leq \delta \leq 1/\sqrt{n} \\
\Theta(1/n^{C_2})
&\text{ if }1/\sqrt{n} < \delta \leq 1/2
\end{cases}
\end{equation*}

\section{Acknowledgments}
We are grateful to the anonymous referees for helpful comments and suggestions.

\bibliographystyle{plain}
\bibliography{mybib}

\section{Appendix A}
In this section we use Fourier analysis of boolean functions. We will use the notation of \cite{num15}.

The following lemma is very similar to a statement from \cite{num9}:
\begin{lemma}
For any function $f$\blue{$:\{0,1\}^n\rightarrow\{0,1\}$} and a parameter $\delta\leq 1/2$ it is the case that:

$$
NS_{\delta}[f] \leq
\delta I[f]
$$
\end{lemma}
\begin{proof}
We will use the Fourier expressions for both of the above (see \cite{num15}):
\begin{align*}
I[f]=\sum_{S} |S| \hat{f}^2(S)
&&
NS_{\delta}[f]=\frac{1}{2} \sum_{S} (1-(1-2 \delta)^{|S|}) \hat{f}^2(S)
\end{align*}
We can now use Bernoulli's inequality $(1-2\delta)^{|S|} \geq 1-2\delta |S|$. Therefore:

$$
NS_{\delta}[f]
\leq
\frac{1}{2} \sum_{S} 2\delta |S| \hat{f}^2(S)
=
\delta
\sum_{S} |S| \hat{f}^2(S)
=
\delta I[f]
$$
This completes the proof.
\end{proof}
\begin{lemma}
\label{NS is increasing}
For a fixed function $f$\blue{$:\{0,1\}^n\rightarrow\{0,1\}$} and for values of $\delta$ satisfying $0 < \delta \leq 1/2$, $NS_{\delta}[f]$ is an increasing function of $\delta$. 
\end{lemma}
\begin{proof}
This follows immediately from the Fourier formula for noise sensitivity.
\end{proof}

\section{Appendix B}
Recall that standard sampling approach requires $O(\frac{1}{NS_{\delta}[f] \epsilon^2})$ queries to estimate noise sensitivity. Here we will show that for sufficiently small constant $\epsilon$, the standard sampling algorithm is optimal up to a constant for all values of $NS_{\delta}[f] \geq 1/2^n$. 

For any $\alpha$ we define $H^{\alpha}$ to be the uniform distribution over all functions $f : \{0,1\}^n \rightarrow \{0,1\}$ for which it is the case that $Pr_{x \in_R \{0,1\}^n}[f(x)=1]=\alpha$.

\begin{lemma}
\label{appendix 2 Lemma 1}
For any sufficiently large $n$, any $\alpha$, satisfying $10^6/2^{n} \leq \alpha \leq 1/2$ and any $\delta$, satisfying $1/n \leq \delta \leq 1/2$, it is the case that:
$$
Pr_{f \in_R H^{\alpha}}
[0.1 \alpha \leq NS_{\delta}[f] \leq 3 \alpha]
\geq 
0.99
$$
\end{lemma}
\begin{proof}
$x=y$ implies that $f(x)=f(y)$, hence:
\begin{equation}
\label{appendix 2 eq 1}
E_{f \in_R H^{\alpha}}[NS_{\delta}[f]]
=
Pr_{f \in_R H^{\alpha}; (x,y) \in_R T_{\delta}}[f(x) \neq f(y)]
=
Pr_{(x,y) \in_R T_{\delta}}[x \neq y]
\cdot
Pr_{f \in_R H^{\alpha}; (x,y) \in_R T_{\delta}}
[f(x) \neq f(y) | x \neq y]
\end{equation}
\label{appendix 2 eq 2}
Since $\delta \geq 1/n$, we have that:
\begin{equation}
Pr_{(x,y) \in_R T_{\delta}}[x \neq y]
=
1-(1-\delta)^n
\geq
1-(1-1/n)^n
\geq
1-e^{-1}
\geq
0.25
\end{equation}
Additionally, we have:

\begin{equation}
\label{appendix 2 eq 3}
Pr_{f \in_R H^{\alpha}; (x,y) \in_R T_{\delta}}
[f(x) \neq f(y) | x \neq y]
=
2 \alpha \cdot \frac{2^n \cdot (1-\alpha)}{2^n-1}
\end{equation}
Combing equations (\ref{appendix 2 eq 1}), (\ref{appendix 2 eq 2}) and (\ref{appendix 2 eq 3}), we get:

\begin{equation}
\label{appendix 2 eq 4}
E_{f \in_R H^{\alpha}}
[NS_{\delta}[f]]
=
Pr_{(x,y) \in_R T_{\delta}}[x \neq y] \cdot
2 \alpha \cdot \frac{2^n \cdot (1-\alpha)}{2^n-1}
\geq
0.20 \alpha
\end{equation}
At the same time, for sufficiently large $n$ we have that:
\begin{equation}
\label{appendix 2 eq 4.5}
E_{f \in_R H^{\alpha}}
[NS_{\delta}[f]]
=
Pr_{(x,y) \in_R T_{\delta}}[x \neq y] \cdot
2 \alpha \cdot \frac{2^n \cdot (1-\alpha)}{2^n-1}
\leq
2.5 \alpha
\end{equation}
Now, we will bound the variance. We have:
\begin{multline}
\label{appendix 2 eq 6}
E_{f \in_R H^{\alpha}}[(NS_{\delta}[f])^2]
=
Pr_{f \in_R H^{\alpha}; (x^1,y^1) \in_R T_{\delta}; (x^2,y^2) \in_R T_{\delta}}[(f(x^1) \neq f(y^1)) \land (f(x^2) \neq f(y^2))]
\\=
(Pr_{(x,y) \in_R T_{\delta}}[x \neq y])^2
\cdot
Pr_{f \in_R H^{\alpha}; (x^1,y^1) \in_R T_{\delta}; (x^2,y^2) \in_R T_{\delta}}[(f(x^1) \neq f(y^1)) \land (f(x^2) \neq f(y^2))\bigg \vert(x^1 \neq y^1) \land (x^2 \neq y^2)]
\end{multline}
We have the following three facts:
\begin{enumerate}
\item For arbitrary $x^1$, $y^1$, $x^2$ and $y^2$, we have 
$$Pr_{f \in_R H^{\alpha}}[(f(x^1) \neq f(y^1)) \land (f(x^2) \neq f(y^2))] \leq Pr_{f \in_R H^{\alpha}}[(f(x^1) \neq f(y^1)] \leq 2 \alpha \cdot 2^n (1-\alpha)/(2^n-1)$$
\item Suppose we further given that no two of $x^1, x^2, y^1$ and $y^2$ are equal to each other. We call this event $K$. If $K$ is the case then $$Pr_{f \in_R H^{\alpha}}[(f(x^1) \neq f(y^1)) \land (f(x^2) \neq f(y^2))]=4
\alpha \cdot \frac{2^n(1-\alpha)}{2^n-1} \cdot \frac{2^n\alpha-1}{2^n-2} \cdot \frac{2^n(1-\alpha)-1}{2^n-3}
$$
\item If $(x^1, y^1)$ and $(x^2, y^2)$ are picked independently from $T_\delta$ conditioned on $x^1 \neq y^1$ and $x^2 \neq y^2$, then $K$ happens unless $x^1=x^2$ or $x^1=y^2$ or $y^1=x^2$ or $y^1=y^2$. By independence and a union bound, the probability of any of these happening is at most $4/2^n$. 
\end{enumerate}
Therefore:
\begin{multline}
\label{appendix 2 eq 7}
Pr_{f \in_R H^{\alpha}; (x^1,y^1) \in_R T_{\delta}; (x^2,y^2) \in_R T_{\delta}}[(f(x^1) \neq f(y^1)) \land (f(x^2) \neq f(y^2)) \bigg \vert(x^1 \neq y^1) \land (x^2 \neq y^2)]
\\ \leq
Pr_{f \in_R H^{\alpha}; (x^1,y^1) \in_R T_{\delta}; (x^2,y^2) \in_R T_{\delta}}[(f(x^1) \neq f(y^1)) \land (f(x^2) \neq f(y^2)) \bigg \vert (x^1 \neq x^2) \land  (x^1 \neq y^1) \land (y^1 \neq x^2) \land (y^1 \neq y^2)
\\ \land (x^1 \neq y^1) \land (x^2 \neq y^2)]
+
Pr_{f \in_R H^{\alpha}; (x^1,y^1) \in_R T_{\delta}; (x^2,y^2) \in_R T_{\delta}}[(f(x^1) \neq f(y^1)) \land (f(x^2) \neq f(y^2)) \bigg \vert (x^1 = x^2) \lor  (x^1 = y^1) \\ \lor (y^1 = x^2) \lor (y^1 = y^2)] \cdot Pr_{(x^1,y^1) \in_R T_{\delta}; (x^2,y^2) \in_R T_{\delta}}[(x^1 = x^2) \lor  (x^1 = y^1) \lor (y^1 = x^2) \lor (y^1 = y^2)]
\\ \leq
4
\alpha \cdot \frac{2^n(1-\alpha)}{2^n-1} \cdot \frac{2^n\alpha-1}{2^n-2} \cdot \frac{2^n(1-\alpha)-1}{2^n-3}
+2 \alpha \cdot 2^n \frac{1-\alpha}{2^n-1}\frac{4}{2^n}
\end{multline}
 Combining this with (\ref{appendix 2 eq 4}) and (\ref{appendix 2 eq 6}) we get a bound for the variance:
\begin{multline}
\label{appendix 2 eq 8}
Var_{f \in_R H^{\alpha}}[NS_{\delta}[f]]
\\ =
(Pr_{(x,y) \in_R T_{\delta}}[x \neq y])^2
\cdot
\left(
4
\alpha \cdot \frac{2^n(1-\alpha)}{2^n-1} \cdot \frac{2^n\alpha-1}{2^n-2} \cdot \frac{2^n(1-\alpha)-1}{2^n-3}
+2 \alpha \frac{2^n \cdot(1-\alpha)}{2^n-1}\frac{4}{2^n}
-
\left(
2 \alpha \cdot \frac{2^n \cdot (1-\alpha)}{2^n-1}
\right)^2
\right)
\\ \leq
\frac{
4 \alpha^2 (1- \alpha)^2
}{(1-1/2^n) \cdot (1-2/2^n) \cdot (1-3/2^n)}
+
\frac{
\alpha(1-\alpha) 
}
{(1-1/2^n)}
\cdot
\frac{8}{2^n}
-
4 \alpha^2 (1-\alpha)^2 
\leq 
\frac{48}{2^n} \alpha^2(1-\alpha)^2
+\frac{16}{2^n} \alpha (1-\alpha)
\leq
\frac{100 \alpha}{2^n}
\end{multline}
(\ref{appendix 2 eq 8}) implies that $NS_{\delta}[f]$ has standard deviation of at most $10\sqrt{\alpha/2^n}$. Since $\alpha \geq 10^6/2^n$, this standard deviation of $NS_{\delta}[f]$ is at most $\alpha/100$. By Chebyshev's inequality, $NS_{\delta}[f]$ is within $\alpha/10$ of its expectation with probability $0.99$. Together with  equations (\ref{appendix 2 eq 4}) and  (\ref{appendix 2 eq 4.5}) this implies the statement of the lemma.
\end{proof}
\begin{theorem}
For any sufficiently large $n$, any $\delta$ satisfying $1/n \leq \delta \leq 1/2$ and any $\alpha_0$ satisfying $10^5/2^{n} \leq \alpha_0 \leq 1/1200$, let $\mathcal{G}$ be an algorithm that given access to a function $f$\blue{$:\{0,1\}^n\rightarrow\{0,1\}$} with probability at least $0.99$ outputs NO if it is given access to a function $f$\blue{$:\{0,1\}^n\rightarrow\{0,1\}$}, satisfying
$$
\alpha_0
\leq
NS_{\delta}[f]
\leq
30 \alpha_0
$$
and with probability at least $0.99$ outputs YES, given a function satisfying:
$$
60
\alpha_0
\leq
NS_{\delta}[f]
\leq
1800 \alpha_0
$$
Then, there is a function $f_0$ given which $\mathcal{G}$ makes $\Omega \left(\frac{1}{\alpha_0}\right)=\Omega \left(\frac{1}{NS_{\delta}[f]}\right)$ queries.
\end{theorem}
\begin{proof}
Consider distributions $H^{10\alpha_0}$ and $H^{600\alpha_0}$. One needs $\Omega \left(\frac{1}{\alpha_0}\right)=\Omega \left(\frac{1}{NS_{\delta}[f]}\right)$ queries to distinguish between them with any constant probability. Both values $10 \alpha_0$ and $600 \alpha_0$ are within the scope of Lemma \ref{appendix 2 Lemma 1}. Therefore, from Lemma \ref{appendix 2 Lemma 1} it is the case that:
$$
Pr_{f \in_R H^{\alpha_0/B_1}}
\left[
\alpha_0
\leq
NS_{\delta}[f]
\leq
30 \alpha_0
\right]
\geq 
0.99
$$
$$
Pr_{f \in_R H^{2B_2\alpha_0/B_1^2}}
\left[
60
\alpha_0
\leq
NS_{\delta}[f]
\leq
1800 \alpha_0
\right]
\geq 
0.99
$$
Since $\mathcal{G}$ is correct with probability at least $0.99$, by a union bound it will distinguish between a random function from $H^{\alpha_0/B_1}$ and $H^{2B_2\alpha_0/B_1^2}$ with probability at least $0.98$. But one needs at least $\Omega \left(\frac{1}{\alpha_0}\right)=\Omega \left(\frac{1}{NS_{\delta}[f]}\right)$ queries to distinguish them. This implies the lower bound on the number of queries $\mathcal{G}$ makes.
\end{proof}
\section{Appendix C}
\subsection{Proof of Lemma \ref{continuity}}
\label{proof 1}
We distinguish three cases:
\begin{enumerate}
\item $l_1 \geq \frac{n}{2}$ 
\item $l_1 \leq \frac{n}{2} \leq l_2$
\item $l_2 \leq \frac{n}{2}$.
\end{enumerate}

We first prove the case 1. Since here ${{n}\choose{l_1}} \geq {{n}\choose{l_2}}$, the left inequality is true. We proceed to prove the right inequality. So, for sufficiently large $n$:

\begin{multline*}
\frac{{{n}\choose{l_1}}}{{{n}\choose{l_2}}}=
\frac{l_2 ! (n-l_2)!}{l_1 ! (n-l_1)!}=
\prod_{i=0}^{l_2-l_1-1} \frac{l_2-i}{n-l_1-i}
\leq
\left(
\frac{l_2}{n-l_2}
\right)^{l_2-l_1}
\leq
\left(
\frac{\frac{n}{2}
+\sqrt{C_1 n \log(n)}}{\frac{n}{2}
-\sqrt{C_1 n \log(n)}}
\right)^{l_2-l_1} \\
\leq
\left(
1+5 \sqrt{\frac{C_1\log(n)}{n}}
\right)^{l_2-l_1}
\leq
\left(
1+5 \sqrt{\frac{C_1\log(n)}{n}}
\right)^{
C_2 \xi \sqrt{\frac{n}{\log(n)}}}
\leq e^{5C_2 \sqrt{C_1} \xi}
=e^{0.5 \xi}
\leq
1+\xi
\end{multline*}
This completes the proof for case 1. We note that this method of bounding the product above was inspired by the proof in \cite{num21}. There it was used to bound probabilities of random walks directly, whereas we are using it to proof this Continuity Lemma first and then apply it later for random walks.

Now, we derive case 3 from it. Suppose $l_1 \leq l_2 \leq \frac{n}{2}$, then, ${{n}\choose{l_1}} \leq {{n}\choose{l_2}}$ which gives us the inequality on the right. To prove the left one, define $l_2'=n-l_1$ and $l_1'=n-l_2$. $l_1'$ and $l_2'$ will satisfy all the requirements for case 1, thus we have:

$$
\frac{{{n}\choose{l_1'}}}{{{n}\choose{l_2'}}}
\leq 1+\xi
$$

This implies:

$$
\frac{{{n}\choose{l_1}}}{{{n}\choose{l_2}}}=
\frac{{{n}\choose{l_2'}}}{{{n}\choose{l_1'}}}
\geq 
\frac{1}{1+\xi}
\geq 1-\xi
$$
This completes the proof for case 3. For case 2, together cases 1 and 3 imply that the following are true:
\begin{align*}
1
\leq
\frac{{{n}\choose{n/2}}}{{{n}\choose{l_2}}}
\leq 1+\xi
&&
1-\xi
\leq
\frac{{{n}\choose{l_1}}}{{{n}\choose{n/2}}}
\leq 1
\end{align*}

Multiplying these together, we show the lemma in case 2.

\subsection{Proof of Lemma \ref{corner cutting}}
\label{proof 2}

Recall that $NS_\delta[f]=2 \cdot Pr_D[f(x)=1 \land f(z)=0]$ and $p_A p_B=Pr_D[f(x)=1 \land f(z)=0 \land \overline{E_1} \land \overline{E_2}]$, which implies the left inequality. We now prove the right one. We have:

\begin{equation}
\label{eq 1-(1)}
\frac{NS_\delta[f]}{2}=
Pr_D[f(x)=1 \land f(z)=0]
\leq
Pr_D[f(x)=1 \land f(z)=0 \land \overline{E_1} \land \overline{E_2}]
+Pr_D[E_1 \lor E_2]
\end{equation}
By Chernoff bound we have:

\begin{equation}
\label{eq 1-(2)}
Pr_D[E_1]
\leq
\exp(
-
\frac{1}{3}
n \cdot \delta 3 
\eta_2 \log n
)
\leq
\frac{1}{n^{\eta_2}}
\end{equation}

Now using the Hoeffding bound together with the fact that since $\delta\leq 1/(\sqrt{n} \log n)$ we have $t_2\leq \sqrt{n}/\log n \cdot (1+3 \eta_2 \log(n)) \leq \sqrt{n \log n}$:
\begin{equation}
\label{eq 1-(3)}
Pr_D[E_2| \overline{E_1}]
\leq
Pr_D[|L(x)-n/2| \geq t_1-t_2]
\leq
2 \exp(
-2(\eta_1-1)^2 \log n
) 
=\frac{2}{n^{2(\eta_1-1)^2}}
\end{equation}
Thus, combining Equations (\ref{eq 1-(1)}) and (\ref{eq 1-(2)}) with Observation \ref{epsilon is large} we have:

\begin{equation}
\label{eq 1-(4)}
Pr_D[E_1 \lor E_2]
\leq
Pr_D[E_1]+
Pr_D[E_2| \overline{E_1}]
\leq
\frac{1}{n^{\eta_2}}+
\frac{1}{n^{2(\eta_1-1)^2}}
\leq \frac{\epsilon}{15 n^C}
\end{equation}
Combining Equations (\ref{eq 1-(1)}) and (\ref{eq 1-(4)}) we get:

$$
\frac{1}{2} NS_{\delta}[f]
\leq
Pr_D[f(x)=1 \land f(z)=0 \land \overline{E_1} \land \overline{E_2}]
+\frac{\epsilon}{15 n^C}
$$
Since $NS_{\delta}[f] \geq \frac{1}{n^{C}}$ and $p_A p_B=Pr_D[f(x)=1 \land f(z)=0 \land \overline{E_1} \land \overline{E_2}]$, this implies the lemma.

\subsection{Proof of Lemma \ref{bound for phi}}
\label{proof 3}
Recall that the probability of any given iteration to be successful is $\phi$. We can upper-bound the probability that the inequality fails to hold by the sum of probabilities the two following bad events: (i) after $(1-\epsilon/16)\cdot 768 \ln 200/(\epsilon^2 \phi)$ iterations there are more than $768 \ln 200/\epsilon^2$ successes. (ii) after $(1+\epsilon/16)\cdot 768 \ln 200/(\epsilon^2 \phi)$ iterations there are less than $768 \ln 200/\epsilon^2$ successes.

By Chernoff bound:
$$
Pr[\text{(i) happens}]
\leq
\exp
\left(
-\frac{1}{3}
\left(\frac{1}{1-\epsilon/16
}-1\right)^2
(1-\epsilon/16)\frac{768 \ln 200}{\epsilon^2}
\right)
\leq 
0.005
$$

$$
Pr[\text{(ii) happens}]
\leq
\exp
\left(
-\frac{1}{2}
\left(1-\frac{1}{1+\epsilon/16 }\right)^2
(1+\epsilon/16)\frac{768 \ln 200}{\epsilon^2}
\right)
\leq 
0.005
$$
This proves the correctness. 

To bound the expected number of iterations, we first observe that from a similar Chernoff bound, after $O(1/(\epsilon^2 \phi))$ iterations with probability at lest $1/2$ we exit the main loop. Each further time we make the same number of iterations, the probability of exiting only increases. This implies that the expected number of iterations is $O(1/(\epsilon^2 \phi))$.

\subsection{Proof of Lemma \ref{bound for p_B}}
\label{proof 4}

Recall that by Equation (\ref{eq 1-(5)}):
$$
p_A p_B
=\sum_{e \in E_I\cap M}
p_{e} q_{e}
$$
By Lemma \ref{lemma edges uniform}:

$$
\left(1-\frac{\epsilon}{310}\right)\frac{\delta}{2^n}\sum_{e \in E_I\cap M} q_{e}
\leq
p_A p_B
\leq
\left(1+\frac{\epsilon}{310}\right)\frac{\delta}{2^n}\sum_{e \in E_I\cap M} q_{e}
$$
Dividing this equation by the equation in Lemma \ref{bound for p_A 1} and substituting $|E_I|=2^{n-1} I[f]$:

$$
\left(1-\frac{\epsilon}{70}\right)\frac{1}{|E_I|}\sum_{e \in E_I\cap M} q_{e}
\leq
p_B
\leq
\left(1+\frac{\epsilon}{70}\right)\frac{1}{|E_I|}\sum_{e \in E_I\cap M} q_{e}
$$

Now applying Observation \ref{M is large} :
\begin{equation}
\label{eq 1-(32)}
\left(1-\frac{\epsilon}{33}\right)\frac{1}{|E_I\cap M|}\sum_{e \in E_I\cap M} q_{e}
\leq
p_B
\leq
\left(1+\frac{\epsilon}{33}\right)\frac{1}{|E_I\cap M|}\sum_{e \in E_I\cap M} q_{e}
\end{equation}

Define $\Psi$ to be the set of pairs of paths $(P_1', P_2')$ for which the following hold:
\begin{itemize}
\item $P_1'$ is a descending path and $P_2'$ is an ascending path.
\item The endpoint of $P_1'$ is the starting point of $P_2'$.
\item The value of $f$ at the starting point of $P_1'$ is one, and it is zero at the endpoint of $P_2'$.
\end{itemize}
If and only if $(P_1,P_2)$ is in $\Psi$, we have that $f(x) \neq f(z)$, therefore these are the only paths contributing to $\phi$.

Using this definition, we have:

\begin{equation}
\label{eq 1-(33)}
\phi
=
\sum_{e \in E_I\cap M}
\left(
Pr_{e' \in_R \mathcal{A}}[e'=e]
\sum_{(P_1', P_2')\in \Psi :e \in P_1'}
Pr_{\mathcal{B}_{e}}[(P_1=P_1') \land (P_2=P_2')]
\right)
\end{equation}

\begin{equation}
\label{eq 1-(34)}
q_{e}
=
\sum_{(P_1', P_2') \in \Psi:e \in P_1'}
Pr_{D}[(P_1=P_1') \land (P_2=P_2')\vert((e \in P_1)) \land \overline{E_1} \land \overline{E_2}]
\end{equation}
Combining Equation (\ref{eq 1-(34)}) and Lemma \ref{lemma 3} we get:

\begin{multline}
\label{eq 1-(35)}
\left(1-\frac{\epsilon}{70}\right)\sum_{(P_1', P_2')\in \Psi:e \in P_1}
Pr_{\mathcal{B}(e)}[(P_1=P_1') \land (P_2=P_2')]
\leq
q_{e}
\\
\leq
\left(1+\frac{\epsilon}{70}\right)\sum_{(P_1', P_2')\in \Psi:e \in P_1}
Pr_{\mathcal{B}(e)}[(P_1=P_1') \land (P_2=P_2')]
\end{multline}

Now, if in Lemma \ref{lemma edge sampler} we fix an $e_2$ and sum over $e_1$ in $E \cap M$ we get that for all $e$:
 
\begin{equation}
\label{eq 1-(36)}
\left(1-\frac{\epsilon}{70}\right)Pr_{e' \in_R \mathcal{A}}[e'=e]
\leq
\frac{1}{|E \cap M|}
\leq
\left(1+\frac{\epsilon}{70}\right)Pr_{e' \in_R \mathcal{A}}[e'=e]
\end{equation}

Combining Equation (\ref{eq 1-(33)}) with Equation (\ref{eq 1-(35)}) and Equation (\ref{eq 1-(36)}) we get:

\begin{equation}
\label{eq 1-(37)}
\left(1-\frac{\epsilon}{33}\right)\phi
\leq
\frac{1}{|E_I\cap M|}\sum_{e \in E_I\cap M} q_{e}
\leq
\left(1+\frac{\epsilon}{33}\right)\phi
\end{equation}

Equations (\ref{eq 1-(32)}) and (\ref{eq 1-(37)}) together imply the lemma.

\section{Appendix D}
\subsection{Proof of Lemma \ref{2-1.5}}
\label{proof 5}
If we make the function equal to one on inputs it possibly equaled zero, the bias of the function cannot decrease.
Additionally, $F'$ can be constructed from $F$ by changing less than $1/n^{C_1}$ fraction of its points. Such a transformation cannot increase the bias by more than $1/n^{C_1}$. This proves (a).

Regarding (b) and (c), we start we the following observation: suppose only one point $f(x^0)$ of an arbitrary function is changed. By union bound the probability for a randomly chosen $x$ that either $x=x^0$ or $x^{\oplus i}=x^0$ is at most $1/2^{n-1}$. Therefore, the probability that $f(x) \neq f(x^{\oplus i})$ cannot change by more than $1/2^{n-1}$. This implies that the influence cannot change by more than $n/2^{n-1}$. 

Regarding noise sensitivity, the situation is analogous. For any $\delta$, the probability that any of $x$ and $y$ equals $x_0$, on which the value of the function is changed, is at most $1/2^{n-1}$ by a union bound. Therefore, $NS_{\delta}[f]$ cannot change by more than $1/2^{n-1}$.

Using the observation and a triangle inequality we get:
$$
\bigl\lvert
I[F]-I[F']
\bigr\rvert
\leq
\frac{2^{n}}{n^{C_1}}
\cdot
\frac{n}{2^{n-1}}
=
\frac{2n}{n^{C_1}}
$$
$$
\bigl\lvert
NS_{\delta}[F]-NS_{\delta}[F']
\bigr\rvert
\leq
\frac{2^{n}}{n^{C_1}}
\cdot
\frac{1}{2^{n-1}}
=
\frac{2}{n^{C_1}}
$$
This completes the proof of (b) and (c).

\subsection{Proof of Lemma \ref{2-1.4}}
\label{proof 6}

Now, we proceed to proving (a). By our condition on $k$, it has to be the case that:
\begin{equation}
\label{eq 2-(0.5)}
Pr_{x \in_R \{0,1\}^n}
[L(x) \geq n/2+k\sqrt{n \log n}]
\geq
\frac{1}{n^{C_1}}
\end{equation}
Then, by Hoeffding's bound, it is the case that:
$$
\frac{1}{n^{C_1}}
\leq
Pr_{x \in_R \{0,1\}^n}
[L(x) \geq n/2+k\sqrt{n \log n}]
\leq
\exp\left(-2n \cdot 
\left( \frac{2 k \sqrt{n \log n}}{n}\right)^2 \right)
=
\exp(-8k^2 \log n)
=\frac{1}{n^{8 k^2}}
$$
And therefore, 
$
k 
\leq
\sqrt{
\frac{C_1}{8}}
$, which proves (a).

Since $n/2+k \sqrt{n \log n}$ is an integer, then $f_0(x)$ equals one if and only if $L(x) \geq n/2+k \sqrt{n \log n}+1$. Therefore, we can rewrite Equation (\ref{eq 2-(0.5)}) as:

\begin{equation}
\label{eq 2-(1)}
B[f_0]+\frac{1}{2^{n}} \cdot {{n}\choose{n/2+k \sqrt{n \log n}}} > \frac{1}{n^{C_1}}
\end{equation}
Additionally, for sufficiently large $n$ we have:
\begin{multline}
\label{eq 2-(2)}
B[f_0]
=
\frac{1}{2^n}
\sum_{l=n/2+k \sqrt{n \log n}+1}^{n}
{{n}\choose{l}} \geq \frac{1}{2^{n}} \cdot {{n}\choose{n/2+k \sqrt{n \log n}+1}} 
\\ =  \frac{1}{2^{n}} \cdot {{n}\choose{n/2+k \sqrt{n \log n}}} \cdot \frac{n-(n/2+k \sqrt{n \log n}+1)+1}{n/2+k \sqrt{n \log n}+1}
= 
\left(
1-O\left(\sqrt{\frac{\log n}{n}} \right)
\right) \cdot \frac{1}{2^n}
{{n}\choose{n/2+k \sqrt{n \log n}}}
\\ \geq
\frac{1}{2} \cdot  \frac{1}{2^n}
{{n}\choose{n/2+k \sqrt{n \log n}}}
\end{multline}
Above, we used the fact that $k$ is at most a constant. Combining Equations (\ref{eq 2-(1)}) and (\ref{eq 2-(2)}) we get that for sufficiently large $n$:
$$
3 B[f_0]
\geq
\frac{1}{n^{C_1}}-\frac{1}{2^n}
{{n}\choose{n/2+k \sqrt{n \log n}}}+2 \cdot \frac{1}{2} \cdot  \frac{1}{2^n}
{{n}\choose{n/2+k \sqrt{n \log n}}}
=
\frac{1}{n^{C_1}}
$$
This together with the fact that $B[f_0] \leq 1/n^{C_1}$, proves (b).

Consider (c) now. We have that:

$$
\frac{1}{3n^{C_1}}
\leq
B[f_0]
=
\frac{1}{2^n}
\sum_{l=n/2+k \sqrt{n \log n}+1}^{n}
{{n}\choose{l}}
\leq
\frac{n}{2^{n}}
{{n}\choose{n/2+k \sqrt{n \log n}+1}}
$$
This implies that:

$$
Pr_{x \in_{R} \{0,1\}^n}
[L(x)=n/2+k \sqrt{n \log n}+1]
=
\frac{1}{2^n}
{{n}\choose{n/2+k \sqrt{n \log n}+1}}
\geq
\frac{1}{3n^{C_1+1}}
$$
At the same time, given that $L(x)=n/2+k \sqrt{n \log n}+1$, if one flips an index $i$ for which $x_i=1$, then it will result that $f_0(x^{\oplus i})=0$. And since the number of such indices is at least half:

\begin{multline*}
I[f_0]
=
n \cdot
Pr_{x \in_R \{0,1\}^n; i \in_R [n]}
[f(x) \neq f(x^{\oplus i})]
\geq
n 
\cdot
Pr_{x \in_{R} \{0,1\}^n}
[L(x)=n/2+k \sqrt{n \log n}+1]
\cdot
\frac{1}{2}
\\ \geq
n 
\cdot
\frac{1}{3n^{C_1+1}}
\cdot
\frac{1}{2}
=
\Omega\left(\frac{1}{n^{C_1}}\right)
\end{multline*}
This proves the left inequality in (c). The right inequality is also correct, because it follows from Lemma \ref{2-1.5} by picking $F$ to be the all-zeros function. Thus, (c) is true.

Regarding noise sensitivity, a known lemma (stated in Appendix A as Lemma \ref{NS is increasing}) implies that noise sensitivity is an increasing function of $\delta$. Therefore, it is enough to consider $\delta=1/n$. Then, for any $x$, if we flip each index with probability $1/n$, the probability that overall exactly one index will be flipped equals $n \cdot \frac{1}{n} (1-1/n)^{n-1}=\Omega(1)$. Additionally, given that only one index is flipped, it is equally likely to be any of the $n$ indices. Therefore, we can lower-bound the noise sensitivity:

\begin{multline*}
NS_{\delta}[f_0]
\geq
NS_{1/n}[f_0]
=
Pr_{(x,y) \in_R T_{1/n}}
[f_0(x) \neq f_0(y)]
\\ \geq
\Omega(1)
\cdot
Pr_{x \in_R \{0,1\}^n, i \in_R [n]}
[f_0(x) \neq f_0(x^{\oplus i})]
=
\Omega(1) \cdot \frac{1}{n} \cdot I[f_0]
\end{multline*}
Together with (c), this implies the left inequality in (d). 

Regarding the right inequality, it follows from Lemma \ref{2-1.5} by picking $F$ to be the all-zeros function. Thus, (d) is true.

\end{document}